\documentclass[12pt]{article}
\usepackage{amssymb}
\usepackage{natbib}
\usepackage{mathrsfs}
\usepackage{amsmath,amssymb}
\usepackage{amsthm}
\usepackage{graphicx,color}
\usepackage{setspace}
\usepackage{footnote}
\usepackage[top=1in, bottom=1in, left=1in, right=1in]{geometry}

\usepackage{subfig}
\usepackage{soul}
\setstcolor{black}
\usepackage{multirow}
\usepackage{hhline}

\usepackage{multirow}

\newcommand{\specialcell}[2][c]{%
  \begin{tabular}[#1]{@{}c@{}}#2\end{tabular}}

\renewcommand{\baselinestretch}{\vv}

\newtheorem{Def}{Definition}[section]

\newtheorem{theorem}{Theorem}
\newtheorem{Prop}{Proposition}
\newtheorem{lemma}{Lemma}
\newtheorem{corollary}{Corollary}

\title{High-dimensional Ordinary Least-squares Projection for Screening Variables}
\author{Xiangyu Wang and Chenlei Leng
\footnote{Wang is a graduate student, Department of Statistical Sciences, Duke
  University (Email: xw56@stat.duke.edu). Leng is Professor, Department of
  Statistics, University of Warwick.   Corresponding author: Chenlei Leng (C.Leng@warwick.ac.uk). We
  thank three referees, an associate editor and Prof. Van Keilegom for their
  constructive comments.}
}
\date{First version: October 4, 2013. This version: \today}
\date{}
\begin{document}
\maketitle
\begin{abstract}
Variable selection is a challenging issue in statistical applications
when the number of predictors $p$ far exceeds the number of observations $n$.
In this ultra-high dimensional setting, the sure independence screening (SIS)
procedure was introduced to significantly reduce the dimensionality by preserving the true model with overwhelming probability, before a refined second stage analysis. However, the aforementioned sure screening property strongly relies
on the assumption that the important variables in the model have large
marginal correlations with the response, which rarely holds in 
reality. To overcome this, we propose a novel and simple
screening technique called the high-dimensional ordinary least-squares 
projection (HOLP). We show that HOLP possesses
the sure screening property and gives consistent variable selection without
the strong correlation assumption, and has a low computational complexity. A ridge type HOLP procedure is also discussed. Simulation
study shows that HOLP performs competitively compared to many other marginal
correlation based methods. An application to a mammalian eye disease data
illustrates the attractiveness of HOLP.

\end{abstract}


\section{Introduction}
The rapid advances of information technology have brought an unprecedented array of large and complex data. In this big data era, a defining feature of a high dimensional dataset is that the  number of variables $p$ far exceeds the number of observations $n$. As a result, the classical ordinary least-squares estimate (OLS) used for linear regression is no longer applicable due to a lack of sufficient degrees of freedom.

Recent years have witnessed an explosion in developing approaches for handling
large dimensional data sets.  A common assumption underlying these approaches is
that although the data dimension is high, the number of the variables
that affect the response is relatively small. The first class
of approaches aim at estimating the parameters and conducting variable
selection simultaneously by penalizing a loss function via a sparsity inducing 
penalty. See, for example, the Lasso
\citep{Tibs:1996,Zhao:Yu:2006,Meinshausen:Buhlmann:2008}, the SCAD
\citep{Fan:Li:2001}, the adaptive Lasso
\citep{Zou:2006,Wang:etal:2007,Zhang:Lu:2007}, the grouped Lasso
\citep{Yuan:Lin:2006}, the LSA estimator \citep{Wang:Leng:2007}, the Dantzig
selector \citep{Candes:Tao:2007}, the bridge regression
\citep{Huang:etal:2008}, and the elastic net
\citep{Zou:Hastie:2005,Zou:Zhang:2009}. However, accurate estimation of a
discrete structure is notoriously difficult. For example, the Lasso can give
non-consistent models if the irrepresentable condition on the design matrix is
violated \citep{Zhao:Yu:2006,Zou:2006}, although computationally more
extensive methods such as those combining subsampling and structure selection \citep{Meinshausen:Buhlmann:2010,Shah:Samworth:2013} may overcome this.

In ultra-high dimensional cases where $p$ is much larger than $n$, these
penalized approaches may not work, and the computation cost for 
large-scale optimization becomes a concern. It is desirable if we
can rapidly reduce the large dimensionality before conducting a 
refined analysis. Motivated by these concerns, \cite{Fan:Lv:2008} initiated a second class of approaches aiming to reduce the dimensionality quickly to a manageable size. In particular, they introduce the sure independence screening (SIS)
procedure that can significantly reduce the dimensionality while preserving
the true model with an overwhelming probability. This important property, termed
the sure screening property, plays a pivotal role for the success of SIS. The
screening operation has been extended, for example, to generalized linear
models \citep{Fan:Fan:2008,Fan:etal:2009,Fan:Song:2010}, additive models
\citep{Fan:etal:2011}, hazard regression
\citep{Zhao:Li:2012,Gorst:Scheike:2013}, and to accommodate conditional
correlation \citep{Barut:etal:2012}. As the SIS builds on marginal
correlations between the response and the features, various extensions of
correlation have been proposed to deal with more general cases
\citep{Hall:Miller:2009,Zhu:etal:2011,Li:etal:2012,LiG:etal:2012}. A number of papers have
proposed alternative ways to improve the marginal correlation aspect of
screening, see, for example,
\cite{Hall:etal:2009,Wang:2009,Wang:2012,Cho:Fryzlewicz:2012}. 

There are two important considerations in designing a screening operator. One
pinnacle consideration is the low computational requirement. After all, screening is predominantly used to quickly reduce the dimensionality. 
The other is that the resulting estimator must possess the sure screening property
under reasonable assumptions. Otherwise, the very purpose of variable
screening is defeated. SIS operates by evaluating the correlations between the
response and one predictor at a time, and retaining the features with top
correlations. Clearly, this estimator can be much more efficiently and easily calculated
than large-scale optimization. For
the sure screening property, a sufficient condition made for SIS
\citep{Fan:Lv:2008} is that the marginal correlations for the important
variables must be bounded away from zero. This condition is referred to as the
marginal correlation condition hereafter. However, for high dimensional
data sets, this assumption is often violated, as predictors
are often correlated. As a result, unimportant variables that are highly
correlated to important predictors will have high priority of being
selected. On the other hand, important variables that are jointly correlated
to the response can be screened out, simply because they are marginally
uncorrelated to the response. Due to these reasons,
\cite{Fan:Lv:2008} put forward an iterative SIS procedure that repeatedly
applies SIS to the current residual in finite many steps. \cite{Wang:2009}
proved that the classical forward regression can also be used for variable
screening,  
and \cite{Cho:Fryzlewicz:2012} advocates a tilting procedure. 

In this paper, we propose a novel variable screener named High-dimensional
Ordinary Least-squares Projection (HOLP), motivated by the ordinary
least-squares estimator and the ridge regression. Like SIS, the resulting HOLP
is straightforward and efficient to compute. Unlike SIS, we show that the sure
screening property holds without the restrictive marginal correlation
assumption. We also discussed Ridge-HOLP, a ridge regression version of HOLP. Theoretically, we prove that the HOLP and Ridge-HOLP possess the sure screening property. More interestingly, we show that both HOLP and Ridge-HOLP are screening consistent in that if we retain a model with the same size as the true model, then the retained model is the same as the true model with probability tending to one. We illustrate the performance of our proposed methods via extensive simulation studies.

The rest of the paper is organized as follows. We elaborate the HOLP estimator
and discuss two viewpoints to motivate it in Section 2. The theoretical properties of HOLP and its ridge version are presented in Section 3. In Section 4,
we use extensive simulation study to compare the HOLP estimator with a number
of competitors and highlight its competitiveness. An analysis of data confirms
its usefulness. Section 5 presents the concluding remarks and discusses future
research. All the proofs are found in the
Supplementary Materials. 

\section{High-dimensional Ordinary Least-Squares Projection}
\subsection{A new screening method}
Consider the familiar linear regression model
\[
y = \beta_1x_1+\beta_2x_2+\cdots+\beta_px_p+\varepsilon,
\]
where $x=(x_1,\cdots,x_p)^T$ is the random predictor vector, $\varepsilon$ is the random error and $y$ is the response. Alternatively, with $n$ realizations of $x$ and $y$, we can write the model as
\[ Y=X\beta+\epsilon,
\]
where $Y \in R^{n}$ is the response vector, $X \in R^{n \times p}$ is the
design matrix, and $\epsilon \in R^n$ consists of {\it i.i.d.} errors. Without
loss of generality, we assume that $\epsilon_i$ follows {\color{black} a distribution
with mean 0 and variance $\sigma^2$}. Furthermore, we assume that $X^TX$ is
invertible when $p<n$ and that $XX^T$ is invertible when $p>n$.   
Denote $\mathcal{M} = \{x_1,...,x_p\}$ as the full model and $\mathcal{M}_S$ as the true model where
$S=\{j: \beta_j\not= 0,~j=1,\cdots,p \}$ is the index set of the nonzero
$\beta_j$'s with cardinality $s = |S|$. To motivate our method, we first look at a general class of {\color{black} linear} estimates of
$\beta$ as
\[ \tilde{\beta}=AY,
\]
where $A \in R^{p\times n}$ maps the response to an estimate and the SIS sets $A=X^T$. Since our emphasis is for screening out the
important variables, $\tilde\beta$ as an estimate of $\beta$ needs not be
accurate but ideally it maintains the rank order of the entries of $|\beta|$
such that the nonzero entries of $\beta$ are large in $\tilde\beta$
relatively. Note that 
\[ AY=A(X\beta+\epsilon)=(AX) \beta+A\epsilon,\]
where $A\epsilon$ consists of linear combinations of zero mean random noises
and $(AX)\beta$ is the signal. In order to preserve the signal part as
much as possible, an ideal choice of $A$ would satisfy that $AX=I$. If this
choice is possible, the signal part would dominate the noise part $A\epsilon$
under suitable conditions. This argument leads naturally to the ordinary
least-squares estimate where $A=(X^TX)^{-1}X^T$ when $p<n$. 

 {However, when $p$ is large than $n$, $X^TX$ is
degenerate and $AX$ cannot be an identity matrix. This fact motivates us to use some kind of generalized inverse of
$X$. 
In Part A of the Supplementary Materials we show that $(X^TX)^{-1}X^T$ can be seen as the
Moore-Penrose inverse of $X$ for $p<n$ and that $X^T(XX^T)^{-1}$ is the
Moore-Penrose inverse of $X$ when $p>n$. We remark that the Moore-Penrose
inverse is one particular form of the
generalized inverse of a matrix. When $A= X^T(XX^T)^{-1}$, $AX$ is no
longer an identity matrix. Nevertheless, as long as $AX$ is diagonally dominant,
$\hat \beta_i ~(i \in S)$ can take advantage of the large diagonal terms of $AX$ to dominate $\hat \beta_i~ (i\not\in S)$ that is just a linear combination of off-diagonal terms.
To show  the diagonal dominance of $AX= X^T(XX^T)^{-1}X$, we
quickly present a comparison to SIS.     
In Fig \ref{fig:motivate}  we plot $AX$ for one
simulated data set with $(n,p)=(50,1000)$, where $X$ is drawn from $N(0,
\Sigma)$ with $\Sigma$ satisfying one of the following: 
(i)~$\Sigma =
I_p$, (ii)~$\sigma_{ij} = 0.6$ and $\sigma_{ii} = 1$, (iii)~$\sigma_{ij} =
0.9^{|i-j|}$ and (iv)~$\sigma_{ij} = 0.995^{|i-j|}$.

 \begin{figure}[!htbp]
   \centering
   \includegraphics[height = 3.8cm]{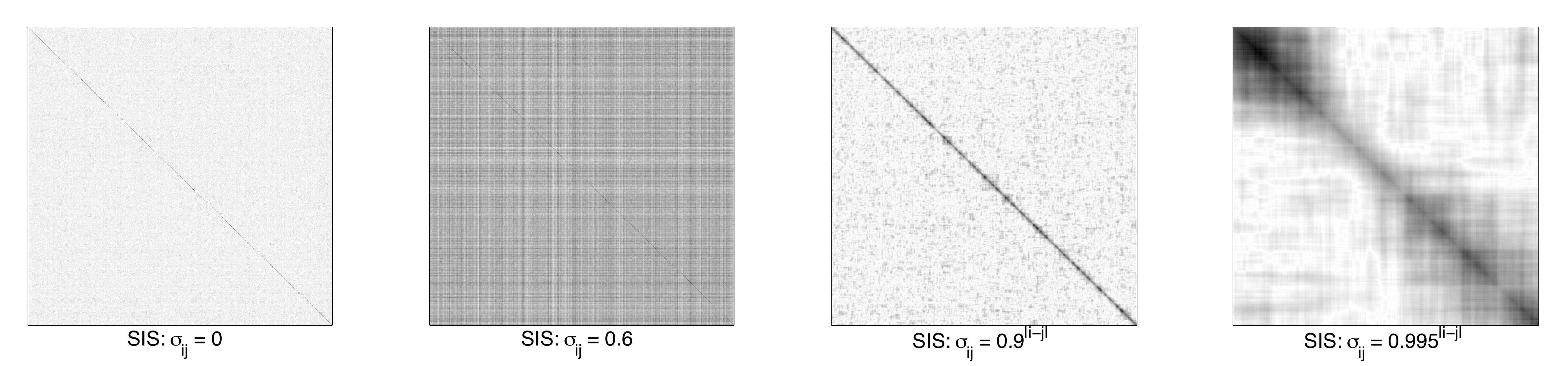}
\medskip
   \centering
   \includegraphics[height = 3.8cm]{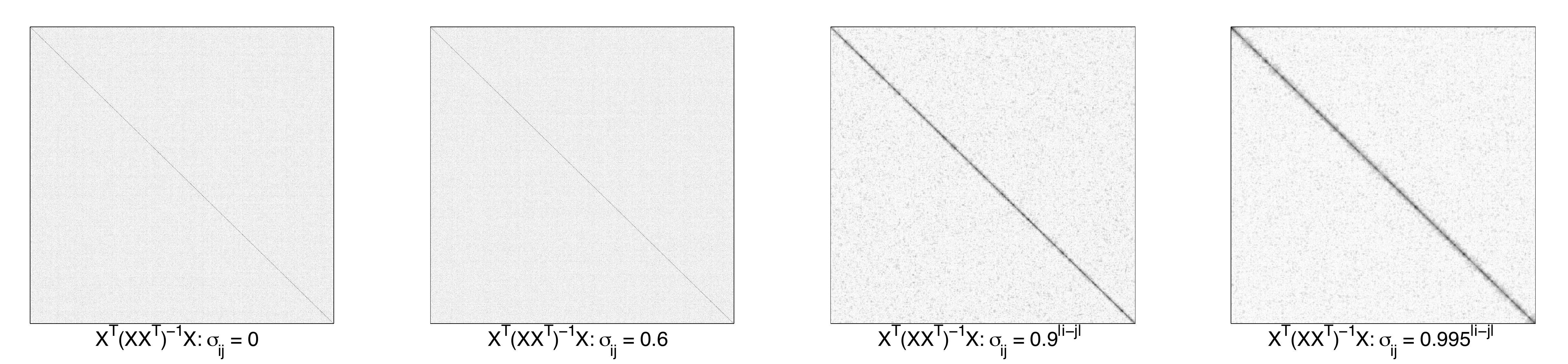}
   \caption{Heatmaps for $AX=X^TX$ in SIS (top) and $AX = X^T(XX^T)^{-1}X$ for
     the proposed method (bottom).}
   \label{fig:motivate}
 \end{figure}

We see a clear pattern of diagonal dominance for $X^T(XX^T)^{-1}X$ under
different scenarios, while the diagonal dominance pattern only emerges {\color{black} for
$AX=X^TX$ in some structures}. To provide an analytical insight, we write
 $X$ via singular value decomposition as $X = VDU^T$,
where $V$ is an $n\times n$ orthogonal matrix, $D$ is an $n\times n$ diagonal
matrix and $U$ is an $p\times n$ matrix that belongs to the Stiefel manifold
$V_{n,p}$. See Part B of the Supplementary Materials for details. Then
  \begin{align*}
    X^T(XX^T)^{-1}X = UU^T,\qquad X^TX = UD^2U^T.
  \end{align*}
Intuitively, $X^T(XX^T)^{-1}X$ reduces the impact from the high correlation of $X$ by removing the random diagonal matrix $D$. As further proved in Part C of the Supplementary Materials, $UU^T$ will be diagonal dominating with overwhelming probability.
 }

These discussions lead to a very simple screening method by first computing
\begin{equation}
\hat\beta = X^T(XX^T)^{-1}Y.\label{eq:holp}
\end{equation}
We name this estimator $\hat\beta$ the High-dimensional Ordinary Least-squares Projection (HOLP) due to the similarity to the classical ordinary least-squares estimate. For variable screening, we follow a very simple strategy by ranking the components of $\hat{\beta}$ and selecting the largest ones. {\color{black} More precisely, let $d$ be the number of the predictors retained after screening.} We choose a submodel $\mathcal{M}_d$ as
\[
\mathcal{M}_d = \{x_j: |\hat \beta_j|\mbox{ are among the largest $d$ of all $|\hat\beta_j|$'s}\}
\quad\mbox{or}\quad
\mathcal{M}_\gamma = \{x_j: |\hat \beta_j|\geq \gamma\}
\]
for some $\gamma$. To see why the HOLP is a projection, we can easily see that
\[
\hat\beta  = X^T(XX^T)^{-1}X\beta+ X^T(XX^T)^{-1}\epsilon,
\]
where the first term indicates that this estimator can be seen as a projection of $\beta$. However, this projection is distinctively different from the usual OLS projection: Whilst the OLS projects the response $Y$ onto the column space of $X$, HOLP uses the row space of $X$ to capture $\beta$. 
We note that many other screening methods, such as tilting and forward
regression, also project $Y$ onto the column space of $X$. Another important
difference between these two projections is the screening mechanism. HOLP
gives a diagonally dominant projection matrix $X^T(XX^T)^{-1}X$, such that the product of this matrix and $\beta$ would be more likely to preserve the rank order of the entries in $\beta$. In contrast, tilting and
forward regression both rely on some goodness-of-fit measure of the selected
variables, aiming to minimize the distance between fitted $\hat Y$ and $Y$. 
An important feature of HOLP is that the matrix $XX^T$ is of full rank whenever $n<p$, in marked contrast to the OLS that is degenerate whenever $n<p$.  Thus, HOLP is unique to high dimensional data analysis from this standpoint.

We now motivate HOLP from a different perspective. Recall the ridge regression
estimate 
\[
\hat\beta(r) = (r I+X^TX)^{-1}X^TY,
\]
where $r$ is the ridge parameter. By letting $r \to \infty$, it is seen that
$r \hat\beta(r) \rightarrow X^TY$. \cite{Fan:Lv:2008} proposed SIS that retains the
large components in $X^TY$ as a way to screen variables. If we let $r \to 0$, the ridge estimator $\hat\beta(r)$ becomes
\[
(X^TX)^{+}X^TY,
\]
where $A^+$ denotes the Moore-Penrose generalized inverse.
 An application of the Sherman-Morrison-Woodbury formula in Part A of the Supplementary Materials gives
\begin{equation*}
(r I+X^TX)^{-1}X^TY = X^T(r I+ XX^T)^{-1}Y. 
\end{equation*}
Then letting $r\rightarrow 0$ gives
\begin{equation*}
(X^TX)^{+}X^TY = X^T(XX^T)^{-1}Y,
\end{equation*}
the HOLP estimator in (\ref{eq:holp}). Therefore, the HOLP estimator can be
seen as the other extreme of the ridge regression estimator by letting {\color{black}$r \to 0$}, as opposed to the marginal screening operator $X^TY$ in \cite{Fan:Lv:2008}
by letting $r \rightarrow \infty$. 
In real data analysis where $X$ and $Y$ are often centered (denoted by $\tilde X$ and $\tilde Y$), the ridge version of HOLP $\tilde X^T(r I+ \tilde X\tilde X^T)^{-1}\tilde Y$ is the correct estimator to use as $\tilde X\tilde X^T$ is now rank-degenerate. Theory on the ridge-HOLP is studied in next section and comparisons with HOLP are provided in the conclusion.

Clearly, HOLP is easy to implement and can be efficiently computed. Its computational complexity is
$O(n^2p)$, while SIS is $O(np)$. 
In the ultra-high dimensional cases where $p\gg n^c$ for any $c$, the
computational complexity of HOLP is only slightly worse than that of
SIS. Another advantage of HOLP is its scale invariance in the signal part
$X^T(XX^T)^{-1}X \beta$. In contrast, SIS is not scale-invariant in  $X^TX\beta$ and its performance may be affected by how
the variables are scaled. 

\section{Asymptotic Properties}\label{theory}
\subsection{Conditions and assumptions}
Recall the linear model
\[
y = \beta_1x_1+\beta_2x_2+\cdots+\beta_px_p+\varepsilon,
\]
where $x=(x_1,\cdots,x_p)^T$ is the random predictor vector, $\varepsilon$ is
the random error and $y$ is the response. In this paper, $X$ denotes
the design matrix. 
Define $Z$ and $z$ respectively as
\begin{equation*}
Z = X\Sigma^{-1/2},~~ z=\Sigma^{-1/2}x,
\end{equation*}
where $\Sigma = cov(x)$ is the covariance matrix of the predictors. For
simplicity, we assume $x_j$'s to have mean $0$ and standard
deviation $1$, i.e, $\Sigma$ is the correlation matrix. It is easy to see that
the covariance matrix of $z$ is an identity matrix. {\color{black}The tail
  behavior of the random error has a significant impact on the screening
  performance. To capture that in a general form, we present the following
  tail condition as a characterization of different distribution families studied in \cite{Vershynin:2010}.
\begin{Def}
  {\bf ($q$-exponential tail condition)} A zero mean distribution $F$ is said to have a q-exponential tail, if any $N\geq 1$ independent random variables $\epsilon_i\sim F$ satisfy that for any $a\in \mathcal{R}^N$ with $\|a\|_2 = 1$, the following inequality holds
  \begin{align*}
    P\bigg(|\sum_{i=1}^N a_i\epsilon_i|>t\bigg)\leq \exp(1-q(t))
  \end{align*}
for any $t>0$ and some function $q(\cdot)$.
\end{Def}
For example, if $\epsilon_i\sim N(0, 1)$, then $\sum_{i=1}^N
a_i\epsilon_i\sim N(0, 1)$. With the classical bound on
the Gaussian tail, one can show that the Gaussian distribution admits a
square-exponential tail in that $q(t)=t^2/2$.
\par This characterization of the tail behavior is an analog of Proposition
5.10 and 5.16 in \cite{Vershynin:2010} and is very general. As shown in
\cite{Vershynin:2010}, we have {\color{black} $q(t) = O(t^2/K^2)$ for some constant $K$ depending on
$F$ if $F$ is sub-Gaussian including Gaussian, Bernoulli, and any bounded
random variables. And we have $q(t) = O(\min\{t/K, t^2/K^2\})$ if $F$ is sub-exponential including exponential, Poisson and $\chi^2$ distribution. Moreover, as shown in \cite{Zhao:Yu:2006}, any random variable satisfies $q(t) = 2k \log t + O(1)$ if it has bounded $2k^{th}$ moments for some positive integer $k$.

 Throughout this paper, $c_i$ and $C_i$ in various places are used to denote positive constants independent of the sample size and the dimensionality.
We make the following assumptions.
\begin{enumerate}
\item[A1.] The transformed $z$ has a spherically symmetric distribution and there exist some $c_1>1$ and $C_1>0$ such that 
\begin{align*}
P\bigg(\lambda_{max}(p^{-1}ZZ^T)>c_1\quad\mbox{or}\quad\lambda_{min}( p^{-1}ZZ^T)<c_1^{-1}\bigg)\leq e^{-C_1n},
\end{align*}
where $\lambda_{max}(\cdot)$ and $\lambda_{min}(\cdot)$ are the largest and
smallest eigenvalues of a matrix respectively. {\color{black} Assume $p>c_0n$ for some $c_0>1$}.
\item[A2.] The random error $\varepsilon$ has mean zero and standard deviation $\sigma$, and is independent of $x$. The standardized error $\varepsilon/\sigma$ has $q$-exponential tails with some function $q(\cdot)$.
\item[A3.] We assume that $var(y) = O(1)$ and that for some $\kappa\geq 0, \nu\geq 0, \tau\geq 0$ and $c_2,c_3,c_4>0$,
\[
\min_{j\in S}|\beta_j|\geq\frac{c_2}{n^\kappa},\quad\quad s\leq c_3n^\nu\quad\mbox{and}\quad \mbox{cond}(\Sigma)\leq c_4n^\tau,
\]
where $\mbox{cond}(\Sigma)=\lambda_{max}(\Sigma)/\lambda_{min}(\Sigma)$ is the conditional number of $\Sigma$.
\end{enumerate}
The assumptions are similar to those in \cite{Fan:Lv:2008} with a key
difference. The strong condition on the marginal correlation between $y$ and
those $x_j$ with $j \in {S}$ required by SIS to satisfy 
\begin{align}
\min_{j\in S} |cov(\beta_j^{-1}y,x_j)|\geq c_5 \label{eq:strong}
\end{align}
for some constant $c_5$, is no longer needed for HOLP.
This marginal correlation condition, as pointed out by \cite{Fan:Lv:2008}, can be easily violated if variables are correlated.
Assumption A1 is similar to but weaker than the concentration property in \cite{Fan:Lv:2008}. See also \cite{Bai:1999}.
 They require all the submatrices of $Z$ consisting of more than $cn$ rows for
 some positive $c$ to satisfy this eigenvalue concentration inequality, while
 here we only require $Z$ itself to hold. {\color{black}The proof in
   \cite{Fan:Lv:2008} can be directly applied to show that A1 is true for the
   Gaussian distribution, and the results in Section 5.4 of \cite{Vershynin:2010} show that the deviation inequality is also true for any sub-Gaussian distribution. It becomes clear later in the proof that the inequality in A1 is not a critical condition for variable screening. In fact, it can be excluded if the model is nearly noiseless.} In A3, $\kappa$ controls the speed at which nonzero $\beta_j$'s decay to 0, $\nu$ is the sparsity rate, and $\tau$ controls the singularity of the covariance matrix.

\subsection{Main theorems}
We establish the important properties of HOLP by presenting three theorems.
\begin{theorem}(Screening property)\label{thm:1}
Assume that A1--A3 hold. If we choose $\gamma_n$ such that 
\begin{align}
  \frac{p\gamma_n}{n^{1-\tau-\kappa}}\rightarrow 0\quad\mbox{ and }\quad\frac{p\gamma_n\sqrt{\log n}}{n^{1-\tau-\kappa}}\rightarrow\infty,\label{eq:gamma1}
\end{align}
then for the same $C_1$ specified in Assumption A1, we have
\[
P\bigg(\mathcal{M}_S\subset\mathcal{M}_{\gamma_n}\bigg) = 1-O\bigg\{\exp\bigg(-C_1\frac{n^{1-2\kappa-5\tau-\nu}}{2\log n}\bigg)\bigg\} - s\cdot \exp\bigg\{1-q\bigg(\frac{\sqrt{C_1}n^{1/2-2\tau-\kappa}}{\sqrt{\log n}}\bigg)\bigg\}.
\]
\end{theorem}
Note that we do not make any assumption on $p$ in Theorem \ref{thm:1} as long
as $p>c_0 n$, allowing $p$ to grow even faster than the exponential rate of the
sample size commonly seen in the literature. The result in Theorem \ref{thm:1} can be of independent interest. If we specialize the dimension to ultra-high dimensional problems,  we have the following
strong results. 
\begin{theorem}(Screening consistency)\label{thm:2}
  In addition to the assumptions in Theorem \ref{thm:1}, if $p$ further satisfies 
\begin{align}
\log p = o\bigg (\min\bigg\{\frac{n^{1-2\kappa-5\tau}}{2\log n}, q\bigg(\frac{\sqrt{C_1}n^{1/2-2\tau-\kappa}}{\sqrt{\log n}}\bigg)\bigg\}\bigg ),\label{eq:p}
\end{align}
then for the same $\gamma_n$ defined in Theorem \ref{thm:1} and the same $C_1$ specified in A1, we have
\begin{align*}
P\bigg(\min_{j\in S}|\hat\beta_j|>&\gamma_n>\max_{j\not\in S}|\hat\beta_j|\bigg)\\
& = 1 - O\bigg\{\exp\bigg(-C_1\frac{n^{1-2\kappa-5\tau-\nu}}{2\log n}\bigg) + \exp\bigg(1-\frac{1}{2}q\bigg(\frac{\sqrt{C_1}n^{1/2-2\tau-\kappa}}{\sqrt{\log n}}\bigg)\bigg)\bigg\}.
\end{align*}
Alternatively, we can choose a submodel $\mathcal{M}_d$ with $d \asymp
n^\iota$ for some $\iota \in (\nu, 1]$ such that
\begin{align*}
P\bigg(\mathcal{M}_S&\subset\mathcal{M}_d\bigg) = 1 - O\bigg\{\exp\bigg(-C_1\frac{n^{1-2\kappa-5\tau-\nu}}{2\log n}\bigg) + \exp\bigg(1-\frac{1}{2}q\bigg(\frac{\sqrt{C_1}n^{1/2-2\tau-\kappa}}{\sqrt{\log n}}\bigg)\bigg)\bigg\}.
\end{align*}
\end{theorem}
The first part of Theorem \ref{thm:2} states that if the number of predictors
satisfies the condition, the important and unimportant variables are separable
by simply thresholding the estimated coefficients in $\hat\beta$. The second part simply states that as long as we
choose a submodel with a dimension larger than that of the true model, we are guaranteed to choose a superset of the variables that contains the true model with probability close to one. If we choose $d=s$, then  
HOLP indeed selects the true model with an overwhelming probability.
This result seems surprising at first glance. It is, however, much weaker than the
consistency of the Lasso under the irrepresentable condition \citep{Zhao:Yu:2006}, as the latter
gives parameter estimation and variable selection at the same time, while our
screening procedure is only used for pre-selecting variables.
\par
When the error $\varepsilon$ follows a sub-Gaussian distribution, HOLP can achieve
screening consistency when the number of covariates increases exponentially
with the sample size.  
\begin{corollary}(Screening consistency for sub-Gaussian errors)
Assume A1--A3. If the standardized error follows a sub-Gaussian distribution, i.e., $q(t) = O(t^2/K^2)$ where $K$ is some constant depending on the distribution, then the condition on $p$ becomes
\begin{align*}
\log p = o\bigg(\frac{n^{1-2\kappa-5\tau}}{\log n}\bigg),
\end{align*}
and for the same $\gamma_n$ defined in Theorem \ref{thm:1} we have
\begin{align*}
P\bigg(\min_{i\in S}|\hat\beta_i| > \gamma_n > \max_{i\not\in S}|\hat\beta_i|\bigg)  = 1 - O\bigg\{\exp\bigg(-C_1\frac{n^{1-2\kappa-5\tau-\nu}}{2\log n}\bigg)\bigg\},
\end{align*}
and with $d \asymp n^\iota$ for some $\iota\in (\nu, 1]$, we have
\begin{align*}
  P\bigg(\mathcal{M}_S\subset\mathcal{M}_d\bigg) = 1 - O\bigg \{\exp\bigg(-C_1\frac{n^{1-2\kappa-5\tau-\nu}}{2\log n}\bigg)\bigg\}.
\end{align*}
\end{corollary}

\par
The next result is an extension of HOLP to the ridge
  regression. Recall the ridge regression estimate
\begin{align*}
\hat\beta(r) = (X^TX+rI_p)^{-1}X^TY = X^T(XX^T+rI_n)^{-1}Y.
\end{align*}
By controlling the diverging rate of $r$, a similar screening property as in Theorem \ref{thm:2} holds for the ridge regression estimate.
\begin{theorem}\label{thm:3}
(Screening consistency for ridge regression) Assume A1--A3 and that $p$
  satisfies \eqref{eq:p}. If the tuning parameter $r$ satisfies 
  $r = o(n^{1-(5/2)\tau-\kappa})$ and in addition to \eqref{eq:gamma1}, $\gamma_n$ further satisfies that $\gamma_n p/(rn^{(3/2)\tau})\rightarrow\infty$,
then for the same $C_1$ in A1, we have
\begin{align*}
P\bigg(\min_{i\in S}|\hat\beta_i(r)|>&\gamma_n>\max_{i\not\in S}|\hat\beta_i(r)|\bigg) \\
&= 1 - O\bigg\{\exp\bigg(-C_1\frac{n^{1-5\tau-2\kappa-\nu}}{2\log n}\bigg)+ \exp\bigg(1-\frac{1}{2}q\bigg(\frac{\sqrt C_1 n^{1/2-2\tau-\kappa}}{\sqrt{\log n}}\bigg)\bigg)\bigg\}.
\end{align*}
With $d \asymp n^\iota$ for some $\iota\in (\nu, 1]$ we have
\begin{align*}
  P\bigg(\mathcal{M}_S\subset\mathcal{M}_d\bigg) = 1 - O\bigg\{\exp\bigg(-C_1\frac{n^{1-2\kappa-5\tau-\nu}}{2\log n}\bigg) + \exp\bigg(-\frac{1}{2}q\bigg(\frac{\sqrt{C_1}n^{1/2-2\tau-\kappa}}{\sqrt{\log n}}\bigg)\bigg)\bigg\}.
\end{align*}
In particular, for any fixed positive constant $r$, the above results hold.
\end{theorem}
Theorem \ref{thm:3} shows that ridge regression can also be used for screening
variables. We recommended to use ridge regression for screening when $XX^T$ is
close to degeneracy or when $n \approx p$. Otherwise, HOLP is suggested due to
its simplicity as it is tuning free. It is also easy to see that the ridge regression estimate has the
same computational complexity as the HOLP estimator. A ridge regression estimator also provides potential for extending the HOLP screening procedure to models other than in linear regression. 
\par
One practical issue for variable screening is how to determine the size of the
submodel. As shown in the theory, as long as the size of the submodel is
larger than the true model, HOLP preserves the non-zero predictors with an
overwhelming probability. Thus, if we can assume $s \asymp n^\nu$ for some $\nu<1$, we can choose a submodel with size $n$, $n-1$  or $n/\log n$ \citep{Fan:Lv:2008,LiG:etal:2012}, or using techniques such as extended BIC \citep{Chen:Chen:2008} to determine the submodel size \citep{Wang:2009}. For simplicity, we mainly use $n$ as the submodel size in numerical study, with some exploration on the extended BIC.

\section{Numerical Studies}\label{simulation}
In this section, we provide extensive numerical experiments to evaluate the performance of HOLP. The structure of this section is organized as
follows. In Part 1, we compare the screening accuracy of HOLP to that
of (I)SIS in \cite{Fan:Lv:2008}, robust rank correlation based screening (RRCS, Li, et al. 2012), the forward regression (FR, Wang, 2009), and
the tilting \citep{Cho:Fryzlewicz:2012}. In Part 2, Theorem 2 and 3 are
numerically assessed under various setups. Because computational complexity is key to a successful screening, in Part 3, we document the
computational time of various methods. Finally, we evaluate the impact
of screening by comparing two-stage procedures where penalized
likelihood methods are employed after screening in Part 4. For implementation, we make use
of the existing R package ``SIS''  and ``tilting'', and write our own code in
R for forward regression.

Although not presented, we have evaluated two additional screeners. The
first is the Ridge-HOLP by setting $r=10$. We
found that the performance is similar to HOLP and therefore report its result only for Part 2. Motivated by the iterative SIS of \cite{Fan:Lv:2008}, we
also investigated an iterative version of HOLP by adding the variable
corresponding to the largest entry in HOLP, one at a time, to the chosen model. In
most cases studied, the screening accuracy of Iterative-HOLP is
similar to or slightly better than HOLP but the computational cost is
much higher. As computation efficiency is one crucial consideration and
also due to the space limit, we
decide not to include the results. 

\subsection{Simulation study I: Screening accuracy}
For simulation study, we set $(p,n) = (1000,100)$ or $(p,n) =
(10000,200)$ and let the random error follow $N(0,\sigma^2)$ with $\sigma^2$
adjusted to achieve different theoretical $R^2$ values defined as
$R^2=var(x^T\beta)/var(y)$ \citep{Wang:2009}.  We use either $R^2=50\%$ for
low or $R^2=90\%$ for high signal-to-noise ratio. 
We simulate covariates from multivariate normal distributions with mean zero
and specify the covariance matrix as the following six models. For each
simulation setup, $200$ datasets are used for $p=1000$ and $100$ datasets are for
$p=10000$. We report the probability of including
the true model 
by selecting a sub-model of size $n$. 
No results
are reported for tilting when $(p,n)=(10000, 200)$ due to its immense computational cost. 

\noindent \textbf{(i) \emph{Independent predictors}}.
This example is from \cite{Fan:Lv:2008} and \cite{Wang:2009} with
$S=\{1,2,3,4,5\}$. We generate $X_i$ from a standard multivariate
normal distribution with independent components. The coefficients are specified as
\[
\beta_i = (-1)^{u_i}(|N(0,1)|+4\log n/\sqrt n), ~\text{where}~u_i\sim Ber(0.4)~\text{for}~i\in S\mbox{ and } \beta_i=0~ \text{for}~i\not\in S.
\]
\par
\noindent \textbf{(ii) \emph{Compound symmetry}}. 
This example is from Example I in \cite{Fan:Lv:2008} and Example 3 in
\cite{Wang:2009}, where all predictors are equally correlated with correlation
$\rho$, and we set $\rho=0.3,~0.6$ or
$0.9$. The coefficients are set to be $\beta_i = 5$ for $i=1,...,5$ and $\beta_i
= 0$ otherwise.

\noindent \textbf{(iii) \emph{Autoregressive correlation}}. 
This correlation structure arises when the predictors are naturally ordered,
for example in time series. The example used here is Example 2 in
\cite{Wang:2009}, modified from the original example in \cite{Tibs:1996}. More
specifically, each $X_i$ follows a multivariate normal distribution, with
$cov(x_i,x_j) = \rho^{|i-j|}$, where $\rho = 0.3, ~0.6$, or $0.9$. The coefficients are specified as
\[
\beta_1 = 3, ~\beta_4 = 1.5, ~\beta_7 = 2,~\text{and}~ \beta_i = 0 \mbox{ otherwise}.
\]

\noindent \textbf{(iv) \emph{Factor models}}. 
Factor models are useful for dimension reduction. Our example is taken from
\cite{Meinshausen:Buhlmann:2010} and 
\cite{Cho:Fryzlewicz:2012}. Let $\phi_j,
j=1,2,\cdots,k$ be independent standard normal random variables. We set
predictors as $x_i = \sum_{j=1}^k \phi_j f_{ij}+\eta_i$, where $f_{ij}$ and
$\eta_i$ are generated from independent standard normal distributions. The
number of the factors is chosen as $k=2, 10$ or $20$ in the simulation while
the coefficients are specified the same as in Example (ii).

\noindent \textbf{(v) \emph{Group structure}}. Group structures depict a
special correlation pattern. 
This example is similar to Example 4 of \cite{Zou:Hastie:2005}, for which we
allocate the 15 true variables into three groups. Specifically, the predictors are generated as
\[
x_{1+3m} = z_1+N(0,\delta^2),~x_{2+3m}= z_2+N(0,\delta^2),~
x_{3+3m} = z_3+N(0,\delta^2),
\]
where $m=0,1,2,3,4$ and $z_i\sim N(0,1)$ are independent. The parameter $\delta^2$ controlling
the strength of the group structure is fixed at 0.01 as in
\cite{Zou:Hastie:2005},  0.05 or 0.1 for a more comprehensive evaluation. The coefficients are set as
\[
\beta_i = 3, ~i=1,2,\cdots,15;~ \beta_i = 0, ~i=16,\cdots,p.
\]

\noindent \textbf{(vi) \emph{Extreme correlation}}. We generate this example
to illustrate the performance of HOLP in extreme cases motivated by the 
challenging Example 4 in \cite{Wang:2009}. As in \cite{Wang:2009}, assuming $z_i\sim N(0,1)$ and $w_i\sim N(0,1)$, we generate the important $x_i$'s as $x_i =
(z_i+w_i)/\sqrt{2}, i=1,2,\cdots,5$ and $x_i = (z_i+\sum_{j=1}^5 w_j)/2,
i=16,\cdots,p$. Setting the coefficients the same as in Example (ii), one can
show that the correlation between the response and the true
predictors is no larger than {\color{black}two thirds} of that between the response
and the false predictors. Thus, the response variable is more
correlated to a large number of unimportant variables. To make the example
even more difficult, we assign another two unimportant predictors to be highly
correlated with each true predictor. Specifically, we let $x_{i+s}, x_{i+2s}
= x_i+ N(0,0.01), ~i=1,2,\cdots,5$. As a result, it will be extremely difficult
to identify any important predictor.

\begin{table}[!htbp]
\footnotesize
  \centering
  \caption{Probability to include the true model when $(p, n) = (10000, 200)$}
  \begin{tabular}{l|l|l|cccccc}
\hline
 &Example &  & HOLP & SIS & RRCS & ISIS & FR & Tilting \\
\hline\hline
\multirow{14}{*}{$\mathcal{R}^2 = 50\%$}&
(i) Independent predictors&  & 0.900 &0.940& 0.890 & 0.620 & 0.570 & \\
\cline{2-9}
&\multirow{3}{*}{(ii) Compound symmetry}
                   & $\rho = 0.3$ & 0.310 & 0.310 & 0.250 & 0.060 & 0.020 &  --- \\
&                  & $\rho = 0.6$ &  0.020 & 0.020 &0.010 & 0.000 & 0.000 &  --- \\
&                  & $\rho = 0.9$ &  0.000 & 0.000 &0.000 & 0.000 & 0.000 &  --- \\
\cline{2-9}
&\multirow{3}{*}{(iii) Autoregressive}& $\rho = 0.3$ & 0.810 & 0.860 & 0.760 & 0.740 & 0.740 &--- \\
&                   & $\rho = 0.6$ & 1.000 & 1.000 & 1.000 & 0.580 & 0.680 &--- \\
&                  & $\rho = 0.9$ & 1.000 & 1.000 & 1.000 & 0.480 & 0.390 &--- \\
\cline{2-9}
&\multirow{3}{*}{(iv) Factor models}& $k = 2$ & 0.450 & 0.010 & 0.010& 0.020 & 0.240 & ---\\
&                      & $k = 10$&  0.050 & 0.000 &0.000& 0.000 & 0.010 & ---\\
&                       & $k = 20$& 0.030 & 0.000 &0.000& 0.000 & 0.000 & ---\\
\cline{2-9}
&\multirow{3}{*}{(v) Group structure} & $\delta^2 = 0.1$ & 1.000 & 1.000 &1.000& 0.000 & 0.000 & ---\\
&                      & $\delta^2 = 0.05$&  1.000 & 1.000 &1.000 & 0.000 & 0.000 & ---\\
&                      & $\delta^2 = 0.01$&  1.000 & 1.000 &1.000 & 0.000 & 0.000 & ---\\
\cline{2-9}
&(vi) Extreme correlation  &                  &  0.580 & 0.000 &0.000 & 0.000 & 0.040 &---\\
\hline
\multirow{14}{*}{$\mathcal{R}^2 = 90\%$}&
(i) Independent predictors &  & 1.000 & 1.000 &1.000 &1.000 & 1.000\\
\cline{2-9}
&\multirow{3}{*}{(ii) Compound symmetry}
                   & $\rho = 0.3$ & 1.000 & 0.820 & 0.710 & 1.000 & 1.000 & ---\\
&                  & $\rho = 0.6$ &  0.960 & 0.550 & 0.320 & 0.420 & 0.960 & ---\\
&                  & $\rho = 0.9$ &  0.100 & 0.030 & 0.000 & 0.000 & 0.000 & --- \\
\cline{2-9}
&\multirow{3}{*}{(iii) Autoregressive}& $\rho = 0.3$ & 0.990 & 0.990&0.980 & 1.000 & 1.000 & ---\\
&                   & $\rho = 0.6$ & 1.000 & 1.000&1.000 & 1.000 & 1.000 &  ---\\
&                  & $\rho = 0.9$ & 1.000 & 1.000 &1.000& 1.000 & 1.000 & ---\\
\cline{2-9}
&\multirow{3}{*}{(iv) Factor model}& $k = 2$ & 0.990 & 0.010 & 0.020 & 0.350 & 0.990 & ---\\
&                      & $k = 10$&  0.850 & 0.000 & 0.000 & 0.060 & 0.700 & ---\\
&                       & $k = 20$& 0.540 & 0.000 & 0.000 & 0.010 & 0.230 & ---\\
\cline{2-9}
&\multirow{3}{*}{(v) Group structure} & $\delta^2 = 0.1$ & 1.000 & 1.000 &1.000 & 0.000 & 0.000 & ---\\
&                      & $\delta^2 = 0.05$& 1.000 & 1.000 &1.000& 0.000 & 0.000 &  ---\\
&                      & $\delta^2 = 0.01$& 1.000 & 1.000 &1.000& 0.000 & 0.000 &  ---\\
\cline{2-9}
&(vi) Extreme correlation  &                  & 1.000 & 0.000 &0.000 &0.000 & 0.210 &--- \\
\hline
  \end{tabular}\label{tab:2}
\end{table}

\noindent \textbf{\emph{Brief summary of the simulation results}}
\par
The results for $(p,n)=(1000,100)$ are shown in Table \ref{tab:1} in the Supplementary Materials and
those for $(p,n)=(10000,200)$ are in Table \ref{tab:2}. We summarize the results in following three points. First, when the signal-to-noise ratio
is low, HOLP, RRCS and SIS outperform ISIS, FR and Tilting in Example (i), (ii), (iii), and (v). For the factor model (iv), neither SIS nor RRCS works while HOLP gives the best performance. In addition, HOLP seems to be the only effective screening method for the extreme correlation model (vi). The poor performance of ISIS, forward regression and tilting in selected scenarios of Example (ii), (iii), and (v) might be caused by the low signal-to-noise ratio, as these methods all depend on the marginal residual deviance that is unreliable when the signal is weak. In particular, they require each true predictor to give the smallest
marginal deviance at some step in order to be selected, {\color{black} imposing a strong condition for achieving satisfactory screening results}. By contrast, SIS, RRCS  and HOLP
select the sub-model in one step and thus eliminate this strong
requirement. The poor performance of SIS and RRCS in Example (iv) and (vi) might be caused by the violation of marginal correlation assumption \eqref{eq:strong} as discussed before. 

Second, when the signal-to-noise ratio increases to $90\%$, significant improvements are seen for all methods. Remarkably, HOLP remains competitive and
achieves an overall good performance. There are occasions where forward
regression and tilting perform slightly better than HOLP, most of which, however, involve only relatively simple structures. The superior
performance of forward regression and tilting under simple structures mainly {\color{black} benefit from their one-at-a-step screening strategy and the high signal-to-noise ratio}. In the simulation study that is not presented here, we also
implemented an iterative version of HOLP, which achieves a similar performance
as forward regression and HOLP in most cases. Yet this strategy fails to a large extent for the group-structured correlation in Example (v). 

Another important feature of HOLP, RRCS and (I)SIS is the flexibility in adjusting the sub-model size. Unlike forward regression and tilting, no limitation is imposed on the sub-model size for HOLP, RRCS and (I)SIS. {\color{black} There might be an advantage to choose a
sub-model of size greater than $n$, so that a better estimation or prediction accuracy can be achieved.}
For example, in Example (ii) when $(p,n,\rho,R^2) = (10000, 200, 0.9,
90\%)$, by selecting $200$ covariates, HOLP preserves the true model
with probability $10\%$. This probability is improved to around
$50\%$ if the sub-model size increases to $1000$, a ten-fold reduction in
dimensionality still. In contrast to HOLP, it is impossible for forward regression and tilting to select a 
sub-model of size larger than $n$ due to {\color{black} the} lack of degrees of freedom. 
\par
{\color{black}
As shown in Section 3,
HOLP relaxes the marginal correlation condition
\eqref{eq:strong} required by SIS. We verify this statement by comparing HOLP and SIS
in a scenario where some important predictors are jointly correlated but
marginally uncorrelated with the response. We take the setup in Example
(ii) with the following model specification
\begin{align*}
  y = 5x_1+5x_2+5x_3+5x_4-20\rho x_5 +\varepsilon.
\end{align*}
It is easy to verify that $cov(x_5, y) = 0$, i.e., $x_5$ is marginally
uncorrelated with $y$. We simulate 200 data sets with $(p, n) = (1000, 100)$
or $(p, n) = (10000, 200)$ with different values of $\rho$. The probability
of including the true model is plotted in Fig \ref{fig:margcomp}. We see
that HOLP performs universally better than SIS for any $\rho$.  
\begin{figure}[!htbp]
  \centering
  \includegraphics[height = 7.6cm]{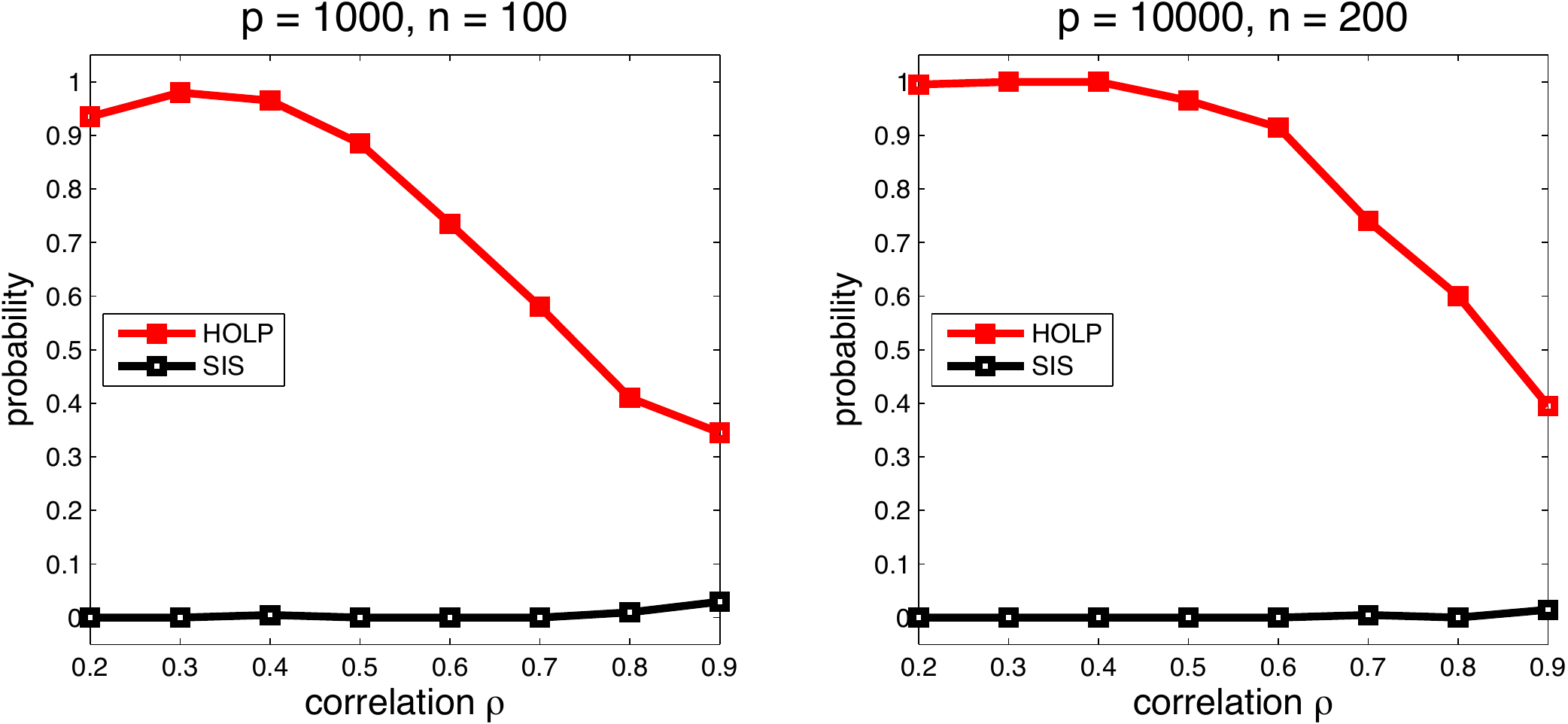}
  \caption{Probability of including the true model for the example where $x_5$ is
    marginally uncorrelated but jointly correlated with $y$.}
  \label{fig:margcomp}
\end{figure}
}
\subsection{Simulation study II: Verification of Theorem 2 and 3}
Theorem 2 and 3 state that HOLP and its ridge regression counterpart are able
to separate the important variables from those unimportant ones with a large
probability, and thus {\color{black} guarantee the effectiveness of variable screening}. In
particular, the two theorems indicate that by choosing a sub-model of size $s$,
we are guaranteed to exactly select the true model. In
this study, we revisit the examples in Simulation I by varying $n, p, s$ to
provide numerical evidences for this claim. Since there are multiple setups, for convenience
we only look at Example (ii),
(iii), (iv) and (v) by fixing the parameters at $\rho = 0.5$, $k = 5$,
$\delta^2 = 0.01$ for $R^2 = 90\%$ and $\rho = 0.3, k = 2, \delta^2 = 0.01$
for $R^2 = 50\%$ respectively. Because Example (vi) is difficult, in
order to demonstrate the two theorems for moderate sample sizes, we relax the
correlation between the {\color{black} important and unimportant} predictors from 0.99 to 0.90 and use a
different growing speed for the number of parameters for this case. To be precise, we set
\begin{align*}
  p =\left\{\begin{aligned}
&4\times\lfloor\exp(n^{1/3})\rfloor & \mbox{ for examples except Example (vi)}\\
&20\times\lfloor\exp(n^{1/4})\rfloor&  \mbox{ for Example (vi)}
\end{aligned}\right.
\end{align*}
and
\begin{align*}
  s = \left\{\begin{aligned}
&1.5\times \lfloor n^{1/4}\rfloor & \mbox{ for } R^2 = 90\%\\
&\lfloor n^{1/4}\rfloor & \mbox{ for } R^2 = 50\%
\end{aligned}\right.,
\end{align*}
where $\lfloor \cdot\rfloor$ is the floor function.
We vary the sample size from 50 to 500 with an increment of 50 and simulate
50 data sets for each example. The probability that $\min_{i\in
  S}|\hat\beta_i|>\max_{i\not\in S}|\hat\beta_i|$ is plotted
in Figure \ref{fig:holp} for HOLP and in Figure \ref{fig:rholp} in Part D of
the Supplementary Materials for the ridge HOLP with $r=10$.

\begin{figure}[!htbp]
  \centering
  \includegraphics[width = 16cm]{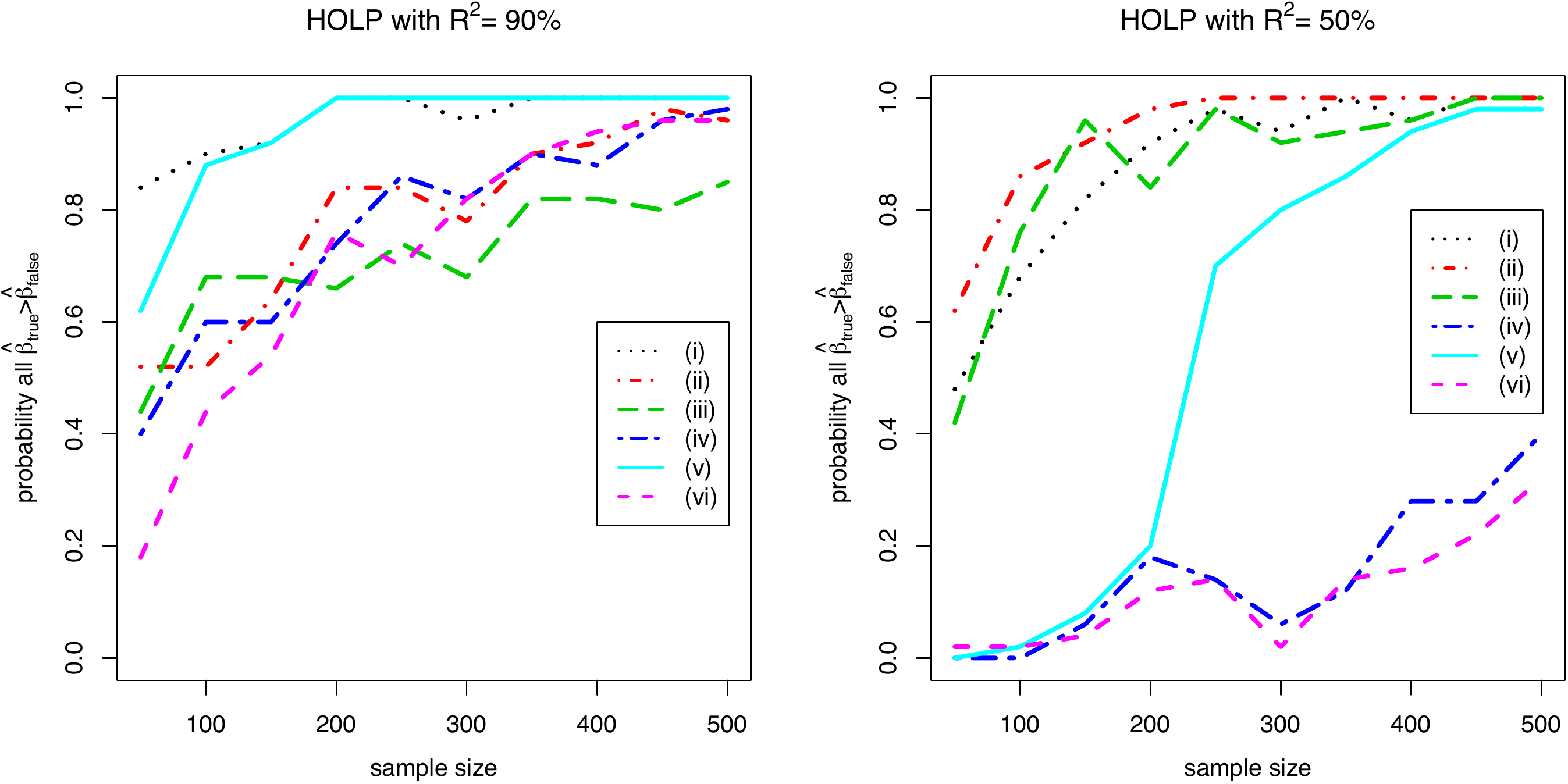}
  \caption{HOLP: $P(\min_{i\in S}|\hat\beta_i|>\max_{i\not\in
      S}|\hat\beta_i|)$ versus the sample size  $n$.}
  \label{fig:holp}
\end{figure}
\par
The increasing trend of the selection probability is explicitly illustrated in Fig \ref{fig:holp}.
Although not plotted,  the
probability for example (vi) when $R^2=50\%$ also tends to one if
the sample size is further increased. 
Thus, we conclude that the probability of correctly identifying the importance rank tends to
one as the sample size increases. A rough exponential pattern can be
recognized from the curves, corresponding to the
rate specified in {\color{black} Corollary 1}. In addition, the probability of identifying the true model is quite similar between HOLP and Ridge-HOLP, echoing the statement
we made at the beginning of Section 4.

\subsection{Simulation study III: Computation efficiency}\label{sect:time}
Computation efficiency is a vital concern for variable screening algorithms, as the primary motivation of screening is to assist variable selection methods, so that they are scalable to large data sets. In this section, we use Example (ii) in Simulation I with
$\rho = 0.9$, $n=100$ and $R^2=90\%$ to illustrate the computation efficiency of
HOLP as compared to SIS, ISIS, forward regression, and tilting. In Figure \ref{fig:1}, we
fix the data dimension at $p=1000$, vary the select sub-model size from 1 to 100, and record the runtime for each
method, while in Figure \ref{fig:3}, we fix the sub-model size at $d=50$ and
vary the data dimension $p$ from 50 to 2500.  Note that the R package 'SIS' computes
$X^TY$ in an inefficient way. For a fair comparison, we write our own code for
computing $X^TY$. Because the computation complexity of tilting is
significantly higher than all other methods, a separate plot excluding tilting is provided for each situation.
\begin{figure}[!htbp]
\begin{center}
\includegraphics[width=16cm]{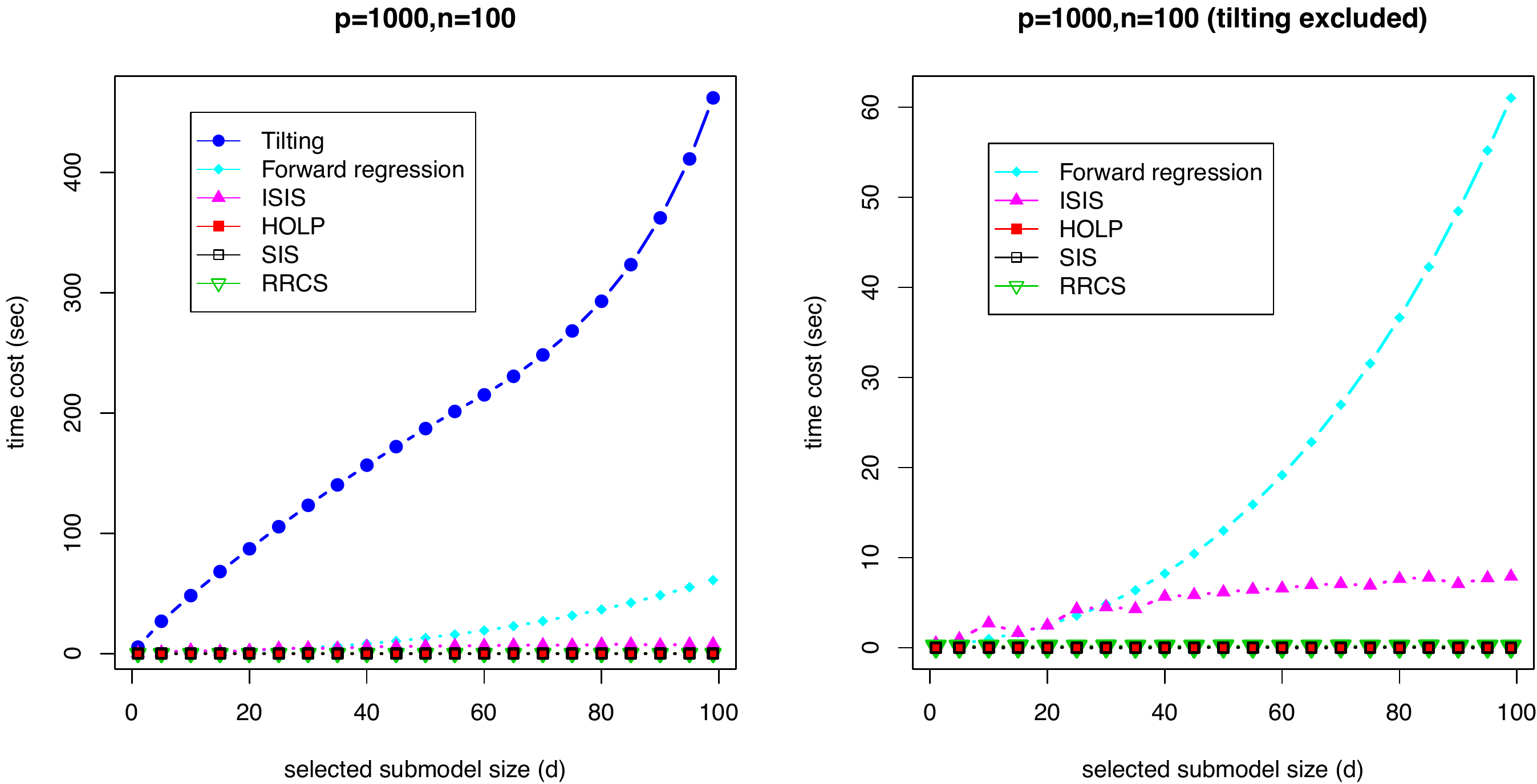}
\caption{Computational time against the submodel size when $(p,n)=(1000,100)$.}
\label{fig:1}
\end{center}
\end{figure}

\begin{figure}[!htbp]
\begin{center}
\includegraphics[width=16cm]{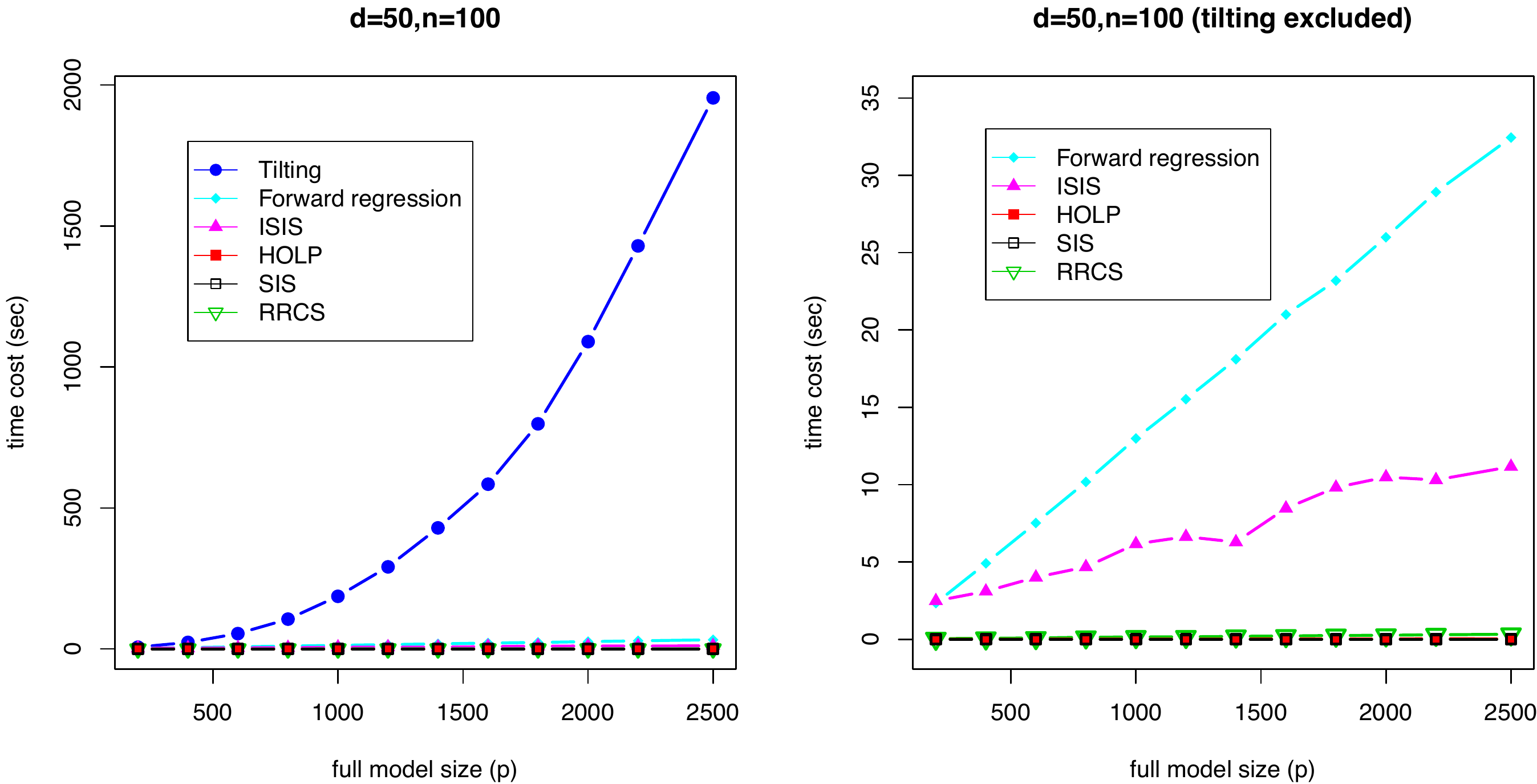}
\caption{Computational time against the total number of the covariates when $(d,n)=(50,100)$. }
\label{fig:3}
\end{center}
\end{figure}

\par
{\color{black}As can be seen from the figures, HOLP, RRCS and SIS are the three most efficient
algorithms. RRCS is actually slightly slower than HOLP and SIS, but not significantly. On the other hand,
tilting demands the heaviest computational cost, followed by
forward regression and ISIS. This result can be interpreted as
follows. When $p$ is fixed as in Figure \ref{fig:1}, HOLP, RRCS and SIS only
incurs a linear complexity on sub-model size $d$,
whereas the complexity of forward regression is approximately quadratic and tilting is $O(k^2d^2 + k^3d)$ where $k$ is the size of active set \citep{Cho:Fryzlewicz:2012}. When $d$ is fixed as in Figure \ref{fig:3}, the
computational time for all methods other than tilting is linearly
increasing on the total number of predictors $p$, while the time for tilting
increasing quadratically with $p$. 
{\color{black}We thus conclude that SIS, RRCS and HOLP are the three preferred
methods in terms of computational complexity. }

\subsection{Simulation study IV: Performance comparison after screening}
Screening as a preselection step aims at assisting the second stage refined analysis on parameter estimation and variable selection. To fully investigate the impact of screening on the second stage analysis, we evaluate and compare different two-stage procedures where screening is followed by variable selection methods such as Lasso or SCAD, as well as these one-stage
variable selection methods themselves. In this section, we look at the six examples in Simulation study I, where the parameters are fixed at $\rho = 0.6, k=10,
\delta^2 = 0.01$ and $\mathcal{R}^2 = 90\%$. {\color{black} To choose the tuning parameter in Lasso or SCAD, we
make use of the extended BIC \citep{Chen:Chen:2008, Wang:2009} to determine a final model that minimizes
\begin{equation*}
EBIC = \log \frac{RSS}{n} + \frac{d}{n}(\log n+2\log p),
\end{equation*}
where $d$ is the number of the predictors in the full model or selected sub-model. 
For all two-stage methods, we first choose a sub-model of size $n$, or use extended BIC to determine the sub-model size (only for HOLP-EBICS), and then apply either Lasso or SCAD to the sub-model to output the final result.
We compare HOLP-Lasso, HOLP-SCAD, HOLP-EBICS (abbreviation for HOLP-EBIC-SCAD) to SIS-SCAD, RRCS-SCAD,
ISIS-SCAD, Tilting, FR-Lasso, FR-SCAD, as well as Lasso and SCAD. The reason
we only apply SCAD to SIS and ISIS is that SCAD is shown to achieve
the best performance in the original paper \citep{Fan:Lv:2008}.
}
\par
Finally, the performance is evaluated for each method in terms of the following
measurements: the number of false negatives (\#FNs, i.e., wrong zeros), the
number of false positives (\#FPs, i.e., wrong predictors), the probability
that the selected model contains the true model (Coverage), the probability
that the selected model is exactly the true model (Exact, i.e., no false
positives or negatives), the estimation error (denoted as
$\|\hat\beta-\beta\|_2$),  the average size of the selected model (Size), and
the algorithm's running time (in seconds per data set).

\par
As in Simulation study I, we simulate 200 data sets for $(p, n) = (1000,
100)$ and 100 data sets for $(p, n) = (10000, 200)$. There will be no results
for tilting in the latter case because of the immense computational cost. The
results for SIS is provided by the package 'SIS', except for the computing
time, which is recorded separately by calculating $X^TY$ directly as discussed
before. All the simulations are run in single thread on PC with an I7-3770
CPU, where we use the package ``glmnet'' for the Lasso and ``ncvreg'' for the SCAD.

\begin{savenotes}
\begin{table}[!htbp]
\scriptsize
  \centering
  \caption{Model selection results for $(p, n) = (10000, 200)$}\label{tab:iv2}
  \begin{tabular}{l|l|ccccccl}
\hline
    example & &\#FNs& \#FPs & Coverage(\%) & Exact(\%) & Size & $||\hat\beta - \beta||_2$& time (sec) \\
\hline\hline
\multirow{9}{*}{\specialcell{(i) Independent\\predictors\\ ~\\ $s = 5, ||\beta||_2 = 3.8$}} 
                             & Lasso      & 0.00  & 0.20 &100.0 &78.0 & 5.20 & 1.21 & 1.15 \\
                             & SCAD       & 0.00  & 0.00 &100.0 &100.0& 5.00 & 0.26 & 18.79\\
                             & ISIS-SCAD  & 0.00  & 0.00 &100.0 &100.0& 5.00 & 0.26 & 211.6\\
                             & SIS-SCAD   & 0.04  & 0.00 &96.0  &96.0 & 4.96 & 0.42 & 0.88\\
                             & RRCS-SCAD  & 0.07  & 0.00 &93.0  &93.0 & 4.93 & 0.53 & 18.12\\
                             & FR-Lasso   &  0.00 & 0.32 &100.0 &78.0 & 5.32 & 1.04 & 5246.6\\
                             & FR-SCAD    &  0.00 & 0.00 &100.0 &100.0& 5.00 & 0.26 & 5247.2\\
                             & HOLP-Lasso & 0.02  & 0.20 &98.0  &78.0 & 5.18 &1.21 &  0.45\\
                             & HOLP-SCAD  &  0.02 & 0.00 &98.0  &98.0 & 4.98 & 0.29 & 0.97\\
                             & HOLP-EBICS  & 0.19  & 0.00 & 82.0 &82.0 &4.81  &0.55  &1.19\\
                             & Tilting &---\\
\hline
\multirow{9}{*}{\specialcell{(ii) Compound\\symmetry\\~\\ $s = 5, ||\beta||_2 = 8.6$ }}     
                             & Lasso      & 1.56  & 2.41  &34.0 & 0.0 &5.85 &9.00 & 1.51\\
                             & SCAD       & 0.01  & 3.65  &99.0 & 6.0 &8.64 & 4.10 &251.5\\
                             & ISIS-SCAD  &1.20   & 5.25  &38.0 & 15.0&9.05 & 7.24 &  465.1\\
                             & SIS-SCAD   &1.51   & 6.19  &26.0 & 19.0&9.68 & 7.97 & 3.84\\
                             & RRCS-SCAD  &1.72   & 6.27  &22.0 & 17.0&9.55 & 8.23 & 25.26\\
                             & FR-Lasso   & 0.14  & 6.89  &86.0 & 0.0 &11.95 & 7.61 & 6904.2\\
                             & FR-SCAD    & 0.20  & 3.35  &85.0 & 6.0 &8.15 & 4.80 & 6909.3\\
                             & HOLP-Lasso &  1.24 & 2.65  &45.0 & 4.0 & 6.41 &8.55 & 0.60\\
                             & HOLP-SCAD  &0.04   & 3.61  &96.0 & 10.0&8.57 & 2.79 & 4.30\\
                             & HOLP-EBICS  & 0.25  & 1.22  &77.0 & 45.0& 4.97 &3.72 & 1.43\\
                             & Tilting&---\\
\hline
\multirow{9}{*}{\specialcell{(iii) Autoregressive\\correlation\\~\\ $s = 3, ||\beta||_2 = 3.9$}}        
                             & Lasso      & 0.00 & 1.06 &100.0 &0.0  & 4.06 & 0.62 & 2.41\\
                             & SCAD       & 0.00 & 0.00 &100.0 &100.0& 3.00 & 0.16 & 34.53\\
                             & ISIS-SCAD  & 0.00 & 0.00 &100.0 &100.0& 3.00 & 0.16 & 342.8\\
                             & SIS-SCAD   & 0.00 & 0.00 &100.0 &100.0& 3.00 & 0.16 & 1.44\\
                             & RRCS-SCAD  & 0.00 & 0.00 &100.0 &100.0& 3.00 & 0.16 & 23.13\\
                             & FR-Lasso   & 0.00 & 1.13 &100.0 &0.0  &4.13 & 0.56 & 10251.2\\
                             & FR-SCAD    & 0.00 & 0.00 &100.0 &100.0&3.00 & 0.16 & 10252.1\\
                             & HOLP-Lasso & 0.00 & 1.12 &100.0 &0.0  & 4.12 & 0.60 & 1.10\\
                             & HOLP-SCAD  & 0.00 & 0.00 &100.0 &100.0& 3.00 & 0.16 & 1.78\\
                             & HOLP-EBICS  & 0.00 & 0.00 &100.0 &100.0& 3.00 & 0.16 & 2.21\\
                             & Tilting&---\\
\hline
\multirow{9}{*}{\specialcell{(iv) Factor Models\\~\\$s = 5, ||\beta||_2 = 8.6$}}     
                             & Lasso      & 4.79 & 6.17  &0.0  &0.0 & 6.38 &11.32 & 1.46\\
                             & SCAD       & 0.11 & 21.08 &91.0 &4.0 &25.97 & 9.41 & 76.30\\
                             & ISIS-SCAD  & 3.09 & 18.06 &3.0  &3.0 &19.97 & 14.27 &409.8\\
                             & SIS-SCAD   & 4.49 & 7.95 &0.0  &0.0 &8.46 & 12.45 &3.34\\
                             & RRCS-SCAD  & 4.47 & 8.16 &0.0  &0.0 &8.69 & 12.50 &25.80\\
                             & FR-Lasso   & 3.54 & 4.45  &13.0 &0.0 &5.91 & 19.40 & 7340.1\\
                             & FR-SCAD    & 1.12 & 21.89 &58.0 &6.0 & 25.77 & 17.18 &7341.8\\
                             & HOLP-Lasso & 3.91 & 6.00  &1.0  &0.0 &7.09 & 11.36 &  0.58\\
                             & HOLP-SCAD  & 0.54 & 14.02 &68.0 &7.0 & 18.48 & 8.83 & 3.00 \\
                             & HOLP-EBICS  & 1.70 & 9.30  &25.0 &10.0& 22.60 & 10.56& 1.69\\ 
                             & Tilting&---\\
\hline
\multirow{9}{*}{\specialcell{(v) Group\\structure\\~\\ $s = 5, ||\beta||_2 = 19.4$}} 
                             & Lasso          & 7.82   & 0.10   &0.0 &0.0 &7.27 & 13.14 &1.51 \\
                             & SCAD           & 11.99  & 115.40&0.0 &0.0 & 118.44 & 25.22&65.67 \\
                             & ISIS-SCAD      & 12.00  & 26.06  &0.0 &0.0 & 29.06 &22.70 &490.4\\
                             & SIS-SCAD       & 11.98  & 21.73  &0.0 &0.0 & 24.75 & 22.68 &2.19\\
                             & RRCS-SCAD      & 11.98  & 21.13  &0.0 &0.0 & 24.15 & 22.77 &20.13\\
                             & FR-Lasso       & 11.75  & 0.89   &0.0 &0.0 & 4.14 & 19.43 & 6916.9\\
                             & FR-SCAD        & 11.96  & 21.50 &0.0 &0.0 & 24.54 & 25.40 &6918.0\\
                             & HOLP-Lasso & 7.75   & 0.11   &0.0 &0.0 & 7.36 & 13.14 &0.62\\
                             & HOLP-SCAD      &  11.98 & 21.95  &0.0 &0.0 & 24.97 & 22.48 & 2.46 \\
                             & HOLP-EBICS      &  11.98 & 0.92   &0.0 &0.0 & 3.94 & 23.23&1.43\\
                             & Tilting&---\\
\hline
\multirow{9}{*}{\specialcell{(vi) Extreme\\correlation\\~\\ $s = 5, ||\beta||_2 = 8.6$}} 
                             & Lasso      & 1.06  & 11.46 & 0.0  &0.0  & 15.40 & 8.60 & 1.34\\
                             & SCAD       & 0.00  & 0.00  & 100.0&100.0& 5.00 & 0.54 & 105.2\\
                             & ISIS-SCAD  & 4.97  & 3.81  & 0.0  &0.0  & 3.85 & 13.18& 507.4\\
                             & SIS-SCAD   & 4.93  & 2.67  & 0.0  &0.0  & 2.74 & 12.10 &3.55\\
                             & RRCS-SCAD  & 5.00  & 2.70  & 0.0  &0.0  & 2.70 & 12.10 &27.75\\
                             & FR-Lasso   & 2.41  & 6.32  & 3.0  &0.0  & 8.89 &10.30 & 7317.6\\
                             & FR-SCAD    &  2.54 & 2.54  & 3.0  &3.0  & 5.00 & 11.21 &7319.2\\
                             & HOLP-Lasso & 0.89  & 10.72 & 42.0 &0.0  & 14.83 & 7.82 &0.43\\
                             & HOLP-SCAD  & 0.00  & 0.00  & 100.0&100.0& 5.00 & 0.54 & 2.70\\
                             & HOLP-EBICS  & 0.70  & 0.70  & 40.0 &40.0 & 5.00 & 2.17 &1.51\\
                             & Tilting&---\\
\hline
  \end{tabular}
\end{table}
\end{savenotes}


Results of the nine methods are shown in Table \ref{tab:iv1} in the Supplementary Materials and Table \ref{tab:iv2}. As can be seen, most methods work well for data sets with relatively simple structures, for example, the independent and
autoregressive correlation structure; likewise, most of them fail for complicated ones, for example, the factor model with 10 factors. The results can be summarized in four main points. First,
HOLP-SCAD achieves the smallest or close to the smallest estimation error for
most cases. Second, SCAD has the overall best coverage probability and the
smallest number of false negatives, followed closely by HOLP-SCAD and FR-SCAD. One potential caveat is, however, the high false positives for SCAD in
many cases. {\color{black}Third, using extended BIC to determine the sub-model size can significantly reduce the false positive rate, although such gain is achieved at the expense of a higher false negative rate and a lower coverage probability. It is also worth noting that using extended BIC can further speed up two-stage methods. Finally, Lasso, HOLP-Lasso, HOLP-SCAD, RRCS-SCAD and SIS-SCAD are the most efficient algorithms in terms of computation.} 

The simulation results suggest that HOLP can not only speed up Lasso and SCAD, but also maintain or even improve their performance in model selection and estimation. In particular, HOLP-SCAD achieves an overal attractive performance. 
We thus conclude that HOLP is an efficient and effective variable
screening algorithm in helping down-stream analysis for parameter estimation and variable selection. 
\subsection{A real data application}
This data set was used to study the mammalian eye diseases by
\cite{Scheetz:etal:2006} where gene expressions on the eye tissues from 120
twelve-week-old male F2 rats were recorded. Among the genes under study, of
particular interest is a gene coded as TRIM32 responsible for causing Bardet-Biedl syndrome \citep{Chiang:etal:2006}.

\par
Following \cite{Scheetz:etal:2006}, we choose 18976 probe sets as they
exhibited sufficient signal for reliable analysis and at least 2-fold
variation in expressions. The intensity values of these genes are evaluated in
the logarithm scale and normalized using the method in
\cite{Irizarry:etal:2003}. Because TRIM32 is believed to be only linked to a
small number of genes, we confine our attention to the top 5000 genes with the
highest sample variance. For comparison, the nine methods in simulation study
IV are examined via 10-fold cross validation {\color{black} and the selected models are refitted via ordinary least squares for prediction purposes}. We report the means and the
standard errors of the mean square errors for prediction and the final
chosen model size in Table \ref{tab:6}. As a reference, we also report these
values for the null model.
\begin{table}[!ht]
\begin{center}
\caption{The 10-fold cross validation error for nine different methods}\label{tab:6}
  \begin{tabular}{l| ccc}
  \hline\hline
  Methods & Mean errors & Standard errors & Final size (median)\\
  \hline
  Lasso& 0.011 & 0.009 & 5\\
  SCAD& 0.015 & 0.011  & 4\\
  \hline
  ISIS-SCAD& 0.012 & 0.006 & 4\\
  SIS-SCAD& 0.010 & 0.004 & 3\\
  RRCS-SCAD& 0.010 & 0.006 & 2\\
  \hline
  FR-Lasso& 0.016 & 0.019 & 4\\
  FR-SCAD & 0.014 & 0.014 & 3\\
  \hline
  HOLP-Lasso&0.012 & 0.006 & 5\\
  HOLP-SCAD& 0.010 & 0.006 & 5\\
  HOLP-EBICS& 0.010 & 0.006 & 5\\
  \hline
  tilting& 0.017 & 0.021 &6\\
  \hline
  NULL   & 0.021 & 0.025 & 0\\
  \hline
  \end{tabular}
\end{center}
\end{table}
\par
{\color{black} From Table \ref{tab:6},  
it can be seen that models selected by HOLP-SCAD, SIS-SCAD and RRCS-SCAD achieve the smallest cross-validation error. It might also be interesting to compare the selected genes by using the full data set, of which a detailed
discussion is provided in Part E and Table \ref{tab:7} in the Supplementary Materials. In particular, gene BE107075 is chosen by all methods other than
tilting. As reported in \cite{Breheny:Huang:2013}, this gene is also selected
via group Lasso and group SCAD.

\section{Conclusion}
In this article, we propose a simple, efficient, easy-to-implement, and flexible method HOLP for screening variables in high dimensional feature space. Compared to other
one-stage screening methods such as SIS, HOLP does not require the strong marginal correlation assumption. Compared to iterative screening methods such as forward regression and tilting, HOLP can be more efficiently computed. Thus, it seems that HOLP holds the two keys at the same time for successful screening: flexible conditions and attractive computation efficiency. Extensive simulation studies show that the performance of HOLP is very competitive, often among the best approaches for screening variables under diverse circumstances with small demand on computational resources. Finally, HOLP is naturally connected to the familiar least-squares estimate for low dimensional data analysis and can be understood as the ridge regression estimate when the ridge parameter goes to zero. 
\par
When $n \approx p$, concerns are raised for the HOLP as $XX^T$ is close to degeneracy. While the screening matrix $X^T(XX^T)^{-1}X = UU^T$ remains diagonally dominant, the noise term $X^T(XX^T)^{-1}\epsilon =UD^{-1}V^T\epsilon$ explodes in magnitude and may dominate the signal, affecting the performance of HOLP. We illustrate this phenomenon via Example (ii) in Section 4.1 with $p$ fixed at 1000 and $\mathcal{R}^2 = 90\%$ for various sample sizes. The probability of including the true model by retaining a sub-model with size $\min\{n, 100\}$ is plotted in Fig \ref{fig:p=n}  (left). 
\begin{figure}[!htbp]
  \centering
  \includegraphics[height=5.3cm]{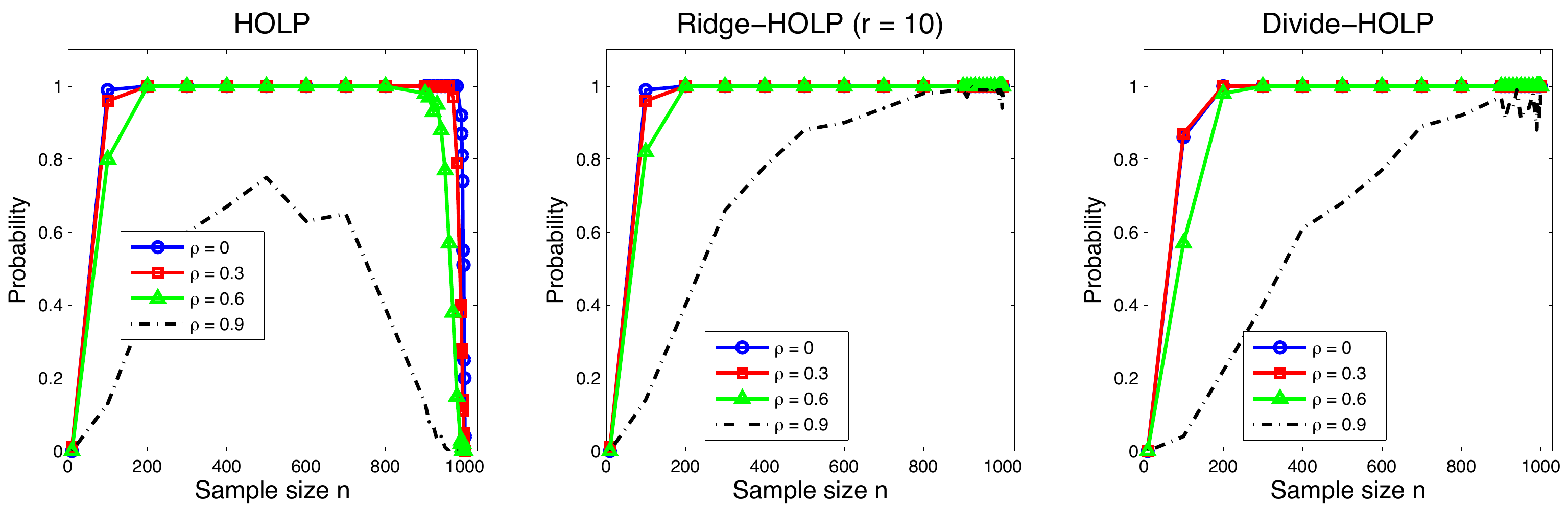}
  \caption{Performance of HOLP, Ridge-HOLP and Divide-HOLP for $p = 1000$.}
  \label{fig:p=n}
\end{figure}
It can be seen that the screening accuracy of HOLP deteriorates whenever $n$ becomes close to $p$. We propose two methods to overcome this issue.
\begin{itemize}
\item 
\textbf{Ridge-HOLP}: As presented in Theorem 3, one approach is to use
Ridge-HOLP by introducing the ridge parameter $r$ to control the explosion of
the noise term. In fact, one can show that $\sigma_{max}(X^T(XX^T +
rI_n)^{-1})\leq r^{-1}\sigma_{max}(X)$, where $\sigma_{max}(X)\approx
O(\sqrt{p} + \sqrt{n}) \approx O(\sqrt{n})$ with large probability. See
\cite{Vershynin:2010}. We verify the performance of Ridge-HOLP via the same
example and plot the result with $r=10$ in Fig \ref{fig:p=n}  (middle).
\item
\textbf{Divide-HOLP}: A second approach is to employ the
``divide-conquer-combine'' strategy, where we randomly partition the data into
$m$ subsets, apply HOLP on each to obtain $m$ reduced models (with a size of
$\min\{n/m, 100/m\}$), and combine the results. This approach ensures
Assumption A1 is satisfied on each subset and can be shown to achieve the same
convergence rate as if the data set were not partitioned. In addition, it
reduces the computational complexity from $O(n^2p)$ to $O(n^2p/m)$. The result
on the same example is shown in Fig \ref{fig:p=n} (right) with $m = 2$. The
performance of Divide-HOLP is on par with Ridge-HOLP when $n$ is close to $p$.
\end{itemize}

\par
There are several directions to further the study on HOLP. First, it is of great interest to extend HOLP to deal with a larger class of models
such as generalized linear models. To address this problem, we may make use of
a ridge regression version of HOLP and study extensions of the results
presented in this paper. Second, we may want to study the screening problem for generalized additive models where nonlinearity is present. 
Third, HOLP may
be used in compressed sensing \citep{Donoho:2006} as in \cite{Xue:Zou:2011} for exactly recovering the
important variables if the sensing matrix satisfies some 
properties. Fourth, we are currently applying the proposed framework for
screening variables in Gaussian graphical models. 
The results will be reported elsewhere.

\section{Acknowledgement}
We thank the three referees, the Associate Editor and the Joint Editor for their constructive comments. Wang’s research was partly supported by grant NIH R01-ES017436 from the National Institute of Environmental Health Sciences.

\renewcommand{\baselinestretch}{1}
\normalsize
\bibliographystyle{apalike}

\newpage 
\begin{center}
{\LARGE Supplementary Materials to "High-dimensional Ordinary Least-squares Projection for Screening Variables"}
\end{center}

\section*{A: Additional results}
\subsection*{The Moore-Penrose inverse}
\begin{Def}
For $A\in R^{m \times n}$, a Moore-Penrose pseudo-inverse of $A$ is defined as a matrix $A^+\in R^{n \times m}$ such that
\[ AA^+A = A, ~A^+AA^+ = A^+, ~(AA^+)^* = AA^+,~(A^+A)^* = A^+A,
\]
where $A^*$ is the conjugate of $A$.
\end{Def}
Using this definition, we can verify $X^+ = X^T(XX^T)^{-1}$ for $p\ge n$ and
$X^+=(X^TX)^{-1}X^T$ for $p\le n$ are both the Moore-Penrose inverse of $X$.

\subsection*{The ridge regression estimator when $r\rightarrow 0$}
Applying the Sherman-Morrison-Woodbury formula
\begin{equation*}
(A+UDV)^{-1} = A^{-1} - A^{-1}U(D^{-1}+VA^{-1}U)^{-1}VA^{-1},
\end{equation*}
we have
\begin{align*}
r(r I_p+X^TX)^{-1} = I_p - X^T(I_n+\frac{1}{r}XX^T)^{-1}X\frac{1}{r} = I_p-X^T(r I_n+XX^T)^{-1}X.
\end{align*}
Multiplying $X^TY$ on both sides, we get
\begin{equation*}
r(r I_p+X^TX)^{-1}X^TY = X^TY-X^T(r I_n+XX^T)^{-1}XX^TY. 
\end{equation*}
The right hand side can be further simplified as
\begin{align*}
\notag X^TY&-X^T(r I_n+XX^T)^{-1}XX^TY \\
\notag &= X^TY-X^T(r I_n+XX^T)^{-1}(r I_n +XX^T - r I_n)Y\\
 &= X^TY-X^TY+r(r I_n+XX^T)^{-1}Y = r X^T(r I_n+XX^T)^{-1}Y. 
\end{align*}
Therefore, we have
\[ (r I_p+X^TX)^{-1}X^TY = X^T(r I_n+XX^T)^{-1}Y.\]

\section*{B: A brief review of the Stiefel manifold}
\par
Let $P\in \mathcal{O}(p)$ be a $p\times p$ orthogonal matrix from the orthogonal group $\mathcal{O}(p)$. Let $H$ denote the first $n$ columns of $P$. Then $H$ is in the Stiefel manifold \citep{Chikuse:2003}. In general, the Stiefel manifold $V_{n,p}$ is the space whose points are $n$-frames in $\mathcal{R}^p$ represented as the set of $p\times n$ matrices $X$ such that $X^TX = I_n$. Mathematically,  we can write
\[
V_{n,p} = \{X \in R^{p\times n}: X^TX = I_n\}.
\]
There is a natural measure $(dX)$ called Haar measure  on the Stiefel manifold, invariant under both right orthogonal and left orthogonal transformations. We standardize it to obtain a probability measure as $[dX] = (dX)/V(n,p)$, where $V(n,p) = {2^n\pi^{np/2}}/{\Gamma_n(1/2p)}$. Let $R_{p,n}$ be the space formed by all $p\times n$ nonsingular matrices. There are several useful results for the distributions on $R_{p,n}$ and $V_{n,p}$, which will be utilized in the following sections.
\begin{lemma} \label{Lemma:1} \citep{Fan:Lv:2008}
An $n\times p$ matrix $Z$ can be decomposed as $Z = VDU^T$ via
the singular value decomposition, where $V\in \mathcal{O}(n), U\in V_{n,p}$
and $D$ is an $n\times n$ diagonal matrix. {\color{black} Let $z_i^T$ denote the
  $i$th row of $Z$, $i=1,2,\cdots,n$. If we assume that $z_i$s are independent and their distribution is invariant under right orthogonal transformation, then the distribution of $Z$ is also invariant under $\mathcal{O}(p)$, i.e,
\[
ZT \stackrel{(d)}{=} Z,~\mbox{for}~ T\in\mathcal{O}(p).
\]}
As a result, we have
\[
U^T \stackrel{(d)}{=} (I_n, 0_{p-n})\times \tilde U,
\]
where $\tilde U$ is uniformly distributed on $\mathcal{O}(p)$. That is, $U$ is uniformly distributed on $V_{n,p}$.
\end{lemma}

Consider a different matrix decomposition. For a $p\times n$ matrix $Z$, define $H_z$ and $T_z$ as
\[
H_z = Z(Z^TZ)^{-1/2},\qquad T_z = Z^TZ.
\]
Then $H_z\in V_{n,p}$ and $Z = H_zT_z^{1/2}$. This is called matrix polar decomposition, where $H_z$ is the orientation of the matrix $Z$. We cite the following result for the polar decomposition.
\begin{lemma} \label{Lemma:2} \citep[Page 41-44]{Chikuse:2003}
Supposed that a $p\times n$ random matrix $Z$ has the density function of the form
\[
f_Z(Z) = |\Sigma|^{-n/2}g(Z^T\Sigma^{-1}Z),
\]
which is invariant under the right-orthogonal transformation of $Z$, where $\Sigma$ is a $p\times p$ positive definite matrix. Then its orientation $H_z$ has the matrix angular central Gaussian distribution (MACG) with a probability density function
\[
MACG(\Sigma) = |\Sigma|^{-n/2}|H_z^T\Sigma^{-1}H_z|^{-p/2}.
\]
In particular, if $Z$ is a $p\times n$ matrix whose distribution is invariant under both the left- and right-orthogonal transformations, then $H_Y$, with $Y = BZ$ for $BB^T = \Sigma$, has the $MACG(\Sigma)$ distribution.
\end{lemma}
When $n=1$, the MACG distribution becomes the angular central Gaussian distribution, a description of the multivariate Gaussian distribution on the unite sphere \citep{Watson:1983}.
\begin{lemma} \label{Lemma:3}
\citep[Page 70, Decomposition of the Stiefel manifold]{Chikuse:2003} 
Let $H$ be a $p\times n$ random matrix on $V_{n,p}$, and write
\begin{align*}
  H = (H_1~ H_2),
\end{align*}
with $H_1$ being a $p\times q$ matrix where $0<q<n$. Then we can write
\begin{align*}
  H_2 = G(H_1)U_1,
\end{align*}
where $G(H_1)$ is any matrix chosen so that $(H_1~ G(H_1))\in \mathcal{O}(p)$; as $H_2$ runs over $V_{n-q,p}$, $U_1$ runs over $V_{n-q, p-q}$ and the relationship is one to one. The differential form $[dH]$ for the normalized invariant measure on $V_{n, p}$ is decomposed as the product
\begin{align*}
  [dH] = [dH_1][dU_1]
\end{align*}
of those $[dH_1]$ and $[dU_1]$ on $V_{q,p}$ and $V_{n-q, p-q}$, respectively.
\end{lemma}

\section*{C: Proofs of the main theory}
\par
The framework of the proof follows \cite{Fan:Lv:2008}, but with many modifications in details. Recall the proposed HOLP screening estimator
\begin{equation*}
\hat\beta = X^T(XX^T)^{-1}Y = X^T(XX^T)^{-1}X\beta+X^T(XX^T)^{-1}\epsilon := \xi+\eta,
\end{equation*}
where $\xi$ can be seen as the signal part and $\eta$ the noise part.
\par
Consider the singular value decomposition of $Z$ as
$Z = VDU^T$, 
where $V\in \mathcal{O}(n), U\in V_{n,p}$ and $D$ is an $n$ by $n$ diagonal
matrix. This gives 
$X = Z\Sigma^{1/2} = VDU^T\Sigma^{1/2}.$ 
Hence, the projection matrix can be written as
\begin{align*}
\notag X^T(XX^T)^{-1}X &= \Sigma^{1/2}UDV^T(VDU^T\Sigma UDV^T)^{-1}VDU^T\Sigma^{1/2} \\
&= \Sigma^{1/2}U(U^T\Sigma U)^{-1}U^T\Sigma^{1/2} := HH^T, \label{eq:defH}
\end{align*}
where
$H = \Sigma^{1/2}U(U^T\Sigma U)^{-1/2}$ 
satisfying $H^TH=I$. In fact, $H$ is the orientation of the
matrix $\Sigma^{1/2}U$. Because $Z$ is sphere symmetrically distributed and
thus invariant under right orthogonal transformation, by Lemma 1, $U$ is
{\color{black} then} uniformly distributed on the Stiefel manifold $V_{n,p}$, 
{\color{black} meaning that} it is invariant under both left- and
right-orthogonal transformation. {\color{black}Therefore,} by Lemma 2, the
matrix $H$ has the MACG($\Sigma$) distribution with regard to the Haar measure
on $V_{n,p}$ as
\[
H\sim |\Sigma|^{-n/2}|H^T\Sigma^{-1}H|^{-p/2},
\]
{\color{black}and} we can write $\xi$ in terms of $H$ as
\[
\xi = HH^T\beta.
\]
The whole proof depends on the properties of $\xi$ and $\eta$, where $\xi$
requires more elaborate analysis. Throughout the whole proof section,
$\|\cdot\|$ denotes the $l_2$ norm of a vector. The following preliminary results are the foundation of the whole theory.
\subsection*{Property of $HH^T\beta$}
In this part, we aim to evaluate the magnitude of $HH^T\beta$. Let $e_i = (0,\cdots,1,0,\cdots,0)^T$ denote the $i$th natural base in the $p$ dimension space {\color{black} and $\tilde e_1$ denote the $n$-dimensional column vector $(1,0,\cdots,0)^T$}. We have the following two lemmas.

\begin{lemma}\label{Lemma:4}
If assumption A1 and A3 hold, for $C>0$ and for any fixed vector $v$ with $\|v\|=1$, there exist constants $c_1',c_2'$ with $0<c_1'<1<c_2'$ such that
\[
P\bigg(v^THH^Tv<c_1'\frac{n^{1-\tau}}{p}\quad\mbox{or}\quad v^THH^Tv>c_2'\frac{n^{1+\tau}}{p}\bigg)<4e^{-Cn}.
\]
In particular for $v=\beta$, whose norm is not 1 though, a
  similar inequality holds for one side with a new $c_2'$ as
\[
P\bigg( \beta^THH^T\beta>c_2'\frac{n^{1+\tau}}{p}\bigg)<2e^{-Cn}.
\]
\end{lemma}

\begin{lemma}\label{Lemma:5}
If assumption A1 and A3 hold, then for any $C>0$, there exists some $c,\tilde c>0$ such that for any $i\in S$,
\begin{equation*}
P\bigg(|e_iHH^T\beta|<c\frac{n^{1-\tau-\kappa}}{p}\bigg)\leq O\bigg\{\exp\bigg(\frac{-Cn^{1-5\tau-2\kappa-\nu}}{2\log n}\bigg)\bigg\},
\end{equation*}
and for any $i\not\in S$,
\begin{equation*}
P\bigg(|e_iHH^T\beta|> \frac{\tilde c}{\sqrt{\log n}}\frac{n^{1-\tau-\kappa}}{p}\bigg)\leq O\bigg\{\exp\bigg(\frac{-Cn^{1-5\tau-2\kappa-\nu}}{2\log n}\bigg)\bigg\},
\end{equation*}
where $\tau,\kappa,\nu$ are the parameters defined in A3.
\end{lemma}

\begin{lemma}\label{Lemma:6}
  Assume A1--A3 hold, we have for any $i\in \{1,2,\cdots,n\}$,
  \begin{align*}
   P\bigg(|\eta_i|>\frac{\sqrt{C_1c_1c_2'c_4}}{\sqrt{\log n}}\frac{n^{1-\kappa-\tau}}{p}\bigg)<\exp\bigg\{1-q\bigg(\frac{\sqrt{C_1}n^{1/2-2\tau-\kappa}}{\sqrt{\log n}}\bigg)\bigg\}+ 3\exp\big(-C_1n\big)
  \end{align*}
where $C_1, c_1, c_4$ are defined in the assumption, and $c_2'$ is defined in Lemma \ref{Lemma:4}.
\end{lemma}
\subsection*{Proof of the three lemmas}
\par
To prove Lemma \ref{Lemma:4}, we need the following two propositions,
{\color{black} first of which is Lemma 3 and the second of which is  similar
  to Lemma 4} in \cite{Fan:Lv:2008}. {\color{black} For completeness, we provide the proof for the second proposition right after the statement.}

\begin{Prop}[Lemma 3 in \cite{Fan:Lv:2008}]
\label{Prop:1}
Let $\xi_i, i=1,2,\cdots,n$ be i.i.d $\chi^2_1$-distributed random variables. Then,
\begin{enumerate}
\item[(i)] for any $\epsilon>0$, we have
\begin{equation*}
  P\bigg(n^{-1}\sum_{i=1}^n\xi_i>1+\epsilon\bigg)\leq e^{-A_\epsilon n}, \label{eq:concentration1}
\end{equation*}
where $A_\epsilon = [\epsilon-\log(1+\epsilon)]/2>0$.
\item[(ii)] for any $\epsilon>0$, we have
\begin{equation*}
  P\bigg(n^{-1}\sum_{i=1}^n\xi_i<1-\epsilon\bigg)\leq e^{-B_\epsilon n},
\end{equation*}
where $B_\epsilon = [-\epsilon-\log(1-\epsilon)]/2>0$.
\end{enumerate}
\par
In other words, for any $C>0$, there exists some $0<c_3'<1<c_4'$ such that
\begin{align*}
  P\bigg(n^{-1}\sum_{i=1}^n\xi_i>c_4'\bigg)\leq e^{-Cn},
\end{align*}
and
\begin{align*}
  P\bigg(n^{-1}\sum_{i=1}^n\xi_i<c_3'\bigg)\leq e^{-Cn},
\end{align*}
\end{Prop}

\begin{Prop}\label{Prop:2}
Let $U$ be uniformly distributed on the Stiefel manifold $V_{n,p}$. Then for any $C>0$, there exist $c_1', c_2'$ with $0<c_1'<1<c_2'$, such that
\begin{equation*}
P\bigg(e_1^TUU^Te_1<c_1'\frac{n}{p}\quad\mbox{or}\quad e_1^TUU^Te_1>c_2'\frac{n}{p}\bigg)\leq 4e^{-Cn}.
\end{equation*}
\end{Prop}

\begin{proof}  First, $U^T$ can be written as $ (I_n~0_{n,p-n})\tilde U$, where $\tilde U$ is uniformly distributed
  on $\mathcal{O}(p)$. Apparently, $\tilde Ue_1$ is uniformly distributed on the
  unite sphere $S^{p-1}$. Thus, letting $\{x_i,i=1,2,\cdots,p\}$ be i.i.d {\color{black}random variables following }$N(0,1)$, we have
\[
\tilde Ue_1\stackrel{(d)}{=} \bigg(\frac{x_1}{\sqrt{\sum_{j=1}^p x_j^2}},\frac{x_2}{\sqrt{\sum_{j=1}^p x_j^2}},\cdots,\frac{x_p}{\sqrt{\sum_{j=1}^p x_j^2}}\bigg)^T.
\]
Hence {\color{black}$U^Te_1$ is the first $n$ coordinates of $\tilde Ue_1$}. It
follows 
\begin{equation*}
e_1^TUU^Te_1\stackrel{(d)}{=}\frac{x_1^2+\cdots+x_n^2}{x_1^2+x_2^2+\cdots+x_p^2}.
\end{equation*}
From Proposition \ref{Prop:1}, we know that for any $C>0$, there exist some $\tilde c_{1}$ and $\tilde c_{2}$ such that
\[
P\bigg(\frac{\sum_{i=1}^n x_i^2}{n}>\tilde c_1\bigg)<e^{-Cn},\qquad P\bigg(\frac{\sum_{i=1}^n x_i^2}{n}<\tilde c_2\bigg)<e^{-Cn},
\]
and
\[
P\bigg(\frac{\sum_{i=1}^p x_i^2}{p}>\tilde c_1\bigg)<e^{-Cp},\qquad P\bigg(\frac{\sum_{i=1}^p x_i^2}{p}<\tilde c_2\bigg)<e^{-Cp}.
\]
Letting $c_1' = \tilde c_2/\tilde c_1, c_2' = \tilde c_1/\tilde c_2$ and by Bonferroni's inequality, we have
\begin{equation*}
P\bigg(e_1^TUU^Te_1<c_1'\frac{n}{p}\quad\mbox{or}\quad e_1^TUU^Te_1>c_2'\frac{n}{p}\bigg)\leq 4e^{-Cn}.
\end{equation*}
The proof is completed.
\end{proof}

\par\noindent
\begin{proof}[\textbf{Proof of Lemma \ref{Lemma:4}}]
Recall the definition of $H$ and
\[
v^THH^Tv = v^T\Sigma^{\frac{1}{2}} U(U^T\Sigma U)^{-1}U^T\Sigma^{\frac{1}{2}} v.
\]
There always exists some orthogonal matrix $Q$ that rotates the vector $\Sigma^{\frac{1}{2}}v$ to the direction of $e_1$, i.e,
\[
\Sigma^{\frac{1}{2}}v = \|\Sigma^{\frac{1}{2}}v\|Qe_1.
\]
Then we have
\[
v^THH^Tv = \|\Sigma^{\frac{1}{2}}v\|^2e_1^TQ^TU(U^T\Sigma U)^{-1}U^TQe_1 = \|\Sigma^{\frac{1}{2}}v\|^2 e_1^T\tilde U(U^T\Sigma U)^{-1}\tilde Ue_1,
\]
where $\tilde U = Q^TU$ is uniformly distributed on $V_{n,p}$, since $U$ is uniformly distributed on $V_{n,p}$ (see discussion in the beginning) and Haar measure is invariant under orthogonal transformation. {\color{black}Now the magnitude of $v^THH^Tv$ can be evaluated in two parts.} For the norm of the vector $\Sigma^{\frac{1}{2}}v$, we have
\begin{equation}
\lambda_{min}(\Sigma)\leq v^T\Sigma v = \|\Sigma^{\frac{1}{2}}v\|^2\leq \lambda_{max}(\Sigma) \label{eq:lemma4},
\end{equation}
and for the rest part,
\[
e_1^T\tilde U(U^T\Sigma  U)^{-1}\tilde Ue_1 \leq \lambda_{max}((U^T\Sigma U)^{-1})\|\tilde Ue_1\|^2\leq \lambda_{min}(\Sigma)^{-1}\|\tilde Ue_1\|^2,
\]
and
\[
e_1^T\tilde U(U^T\Sigma U)^{-1}\tilde Ue_1 \geq \lambda_{min}((U^T\Sigma U)^{-1})\|\tilde Ue_1\|^2\geq \lambda_{max}(\Sigma)^{-1}\|\tilde Ue_1\|^2.
\]
Consequently, we have
\begin{equation}
v^THH^Tv\leq \frac{\lambda_{max}(\Sigma)}{\lambda_{min}(\Sigma)}e_1^TUU^Te_1,\qquad v^THH^Tv\geq \frac{\lambda_{min}(\Sigma)}{\lambda_{max}(\Sigma)}e_1^TUU^Te_1.\label{eq:lemma4.1}
\end{equation}
Therefore, following Proposition \ref{Prop:2} and A3, for any $C>0$ we have
\[
P\bigg(v^THH^Tv< c_1'c_4\frac{n^{1-\tau}}{p}\quad\mbox{or}\quad v^THH^Tv> c_2'c_4^{-1}\frac{n^{1+\tau}}{p}\bigg)\leq 4e^{-Cn}.
\]
Denoting $c_1'c_4$ by $c_1'$ and $c_2'c_4^{-1}$ by $c_2'$, we obtain the equation in the lemma.
\par
{\color{black}Next for $v=\beta$, it follows from Assumption A3 that
\begin{equation}
var(Y) =\beta^T\Sigma\beta+\sigma^2 = O(1). \label{eq:sigma1}
\end{equation}
{\color{black}Equation \eqref{eq:lemma4} then can be updated as
\[
\beta^T \Sigma \beta \leq c'
\]
for some constant $c'$, and \eqref{eq:lemma4.1} now becomes
\begin{equation*}
\beta^THH^T\beta\leq \frac{c'}{\lambda_{min}(\Sigma)}e_1^TUU^Te_1.
\end{equation*}
Since the trace of the covariance matrix $\Sigma$ is $p$, which entails that $\lambda_{max}(\Sigma)\geq 1$ and $\lambda_{min}(\Sigma)\leq 1$. Now with assumption A3, we have
\begin{equation}
 \lambda_{min}(\Sigma)\geq \frac{\lambda_{min}(\Sigma)}{\lambda_{max}(\Sigma)}> c_4^{-1}n^{-\tau}.\label{eq:lambdamin}
\end{equation}
Combining the above two equations, we have that for some new $c_2'>0$, it
holds 
\begin{equation*}
P\left(\beta^THH^T\beta>c_2'\frac{n^{1+\tau}}{p}\right)<2e^{-Cn}.
\end{equation*}}}
\end{proof}
\par

The proof of Lemma \ref{Lemma:5} relies on the results from Stiefel manifold. We first prove following propositions, which can assist the proof of Lemma \ref{Lemma:5}.

\begin{Prop}\label{prop:decomp}
  Assume a $p\times n$ matrix $H\in V_{n,p}$ follows the Matrix Angular Central Gaussian distribution with covariance matrix $\Sigma$. From Lemma \ref{Lemma:3} we can decompose $H = (T_1, H_2)$ with $T_1 = G(H_2)H_1$, where $H_2$ is a $p\times (n-q)$ matrix, $H_1$ is a $(p-n + q)\times q$ matrix and $G(H_2)$ is a matrix such that $(G(H_2), H_2) \in \mathcal{O}(p)$. We have following result
  \begin{align}
    H_1|H_2 \sim MACG(G(H_2)^T\Sigma G(H_2))
  \end{align}
with regard to the invariant measure $[H_1]$ on $V_{q, p-n+q}$.
\end{Prop}

\begin{proof}
  Recall that $H$ follows a $MACG(\Sigma)$ on $V_{n,p}$,which possesses a density as
\begin{align*}
  p(H) \propto |H^T\Sigma^{-1}H|^{-p/2}[dH].
\end{align*}
Using the identity for matrix determinant
\begin{align*}
  \begin{vmatrix}
    A & B\\
    C & D
  \end{vmatrix}
 = |A||D - CA^{-1}B| = |D||A - BD^{-1}C|,
\end{align*}
we have
\begin{align*}
  P(H_1,H_2)& \propto |H_2^T\Sigma^{-1}H_2|^{-p/2}(T_1^T\Sigma^{-1}T_1 - T_1^T\Sigma^{-1}H_2(H_2^T\Sigma^{-1}H_2)^{-1}H_2^T\Sigma^{-1}T_1)^{-p/2}\\
 &= |H_2^T\Sigma^{-1}H_2|^{-p/2}(H_1^TG(H_2)^T(\Sigma^{-1} - \Sigma^{-1}H_2(H_2^T\Sigma^{-1}H_2)^{-1}H_2^T\Sigma^{-1})G(H_2)H_1)^{-p/2}\\
& = |H_2^T\Sigma^{-1}H_2|^{-p/2}(H_1^TG(H_2)^T\Sigma^{-1/2}(I - T_2)\Sigma^{-1/2}G(H_2) H_1)^{-p/2},
\end{align*}
where $T_2 = \Sigma^{-1/2}H_2(H_2^T\Sigma^{-1} H_2)^{-1}H_2^T\Sigma^{-1/2}$ is an orthogonal projection onto the linear space spanned by the columns of $\Sigma^{-1/2}H_2$. It is easy to verify the following result by using the definition of $G(H_2)$,
\begin{align*}
  [\Sigma^{1/2}G(H_2)(G(H_2)^T\Sigma G(H_2))^{-1/2}, ~\Sigma^{-1/2}H_2(H_2^T\Sigma^{-1} H_2)^{-1/2}] \in \mathcal{O}(p),
\end{align*}
and therefore we have
\begin{align*}
  I - T_2 = \Sigma^{1/2}G(H_2)(G(H_2)^T\Sigma G(H_2))^{-1}G(H_2)^T\Sigma^{1/2},
\end{align*}
which simplifies the density function as
\begin{align*}
  P(H_1,H_2)\propto |H_2^T\Sigma^{-1}H_2|^{-p/2}(H_1^T (G(H_2)^T\Sigma G(H_2))^{-1} H_1)^{-p/2}.
\end{align*}
Now it becomes clear that $H_1|H_2$ follows the Matrix Angular Central Gaussian distribution $ACG(\Sigma')$, where
\begin{align*}
  \Sigma' = G(H_2)^T\Sigma G(H_2).
\end{align*}
This completes the proof.
\end{proof}

\begin{Prop}\label{prop:offdiag}
  Assume $H\in V_{n,p}$. Write $H = (T_1, H_2)$ where $T_1 = (T_1^{(1)},T_1^{(2)}, \cdots, T_1^{(p)})^T$ is the first column of $H$, then we have
  \begin{align*}
     e_1^THH^Te_2 \stackrel{(d)}{=} \quad T_1^{(1)}T_1^{(2)}~ \bigg| ~T_1^{(1)2} = e_1^THH^Te_1.
  \end{align*}
\end{Prop}
\begin{proof}
  Notice that for any orthogonal matrix $Q\in\mathcal{O}(n)$, we have
\begin{align*}
  e_1^THH^Te_2 = e_1^THQQ^TH^Te_2 = e_1^TH'H^{'T}e_2.
\end{align*}
Write $H' = HQ = (T_1', H_2')$, where $T_1' = [T_1^{'(1)}, T_1^{'(2)},\cdots, T_1^{'(p)}], ~H_2' = [H_2^{'(i,j)}]$. If we choose $Q$ such that the first row of $H_2'$ are all zero (this is possible as we can choose the first column of $Q$ being the first row of $H$ upon normalizing), i.e.,
\begin{align*}
  e_1^TH' = [T_1^{'(1)},~0,\cdots,0]\qquad e_2^TH' = [T_1^{'(2)},~ H_2^{'(2,1)},\cdots,~H_2^{'(2,n-1)}],
\end{align*}
then immediately we have $e_1^THH^Te_2 = e_1^TH'H^{'T}e_2 = T_1^{'(1)}T_1^{'(2)}$. This indicates that
\begin{align*}
  e_1^THH^Te_2 \stackrel{(d)}{=} \quad T_1^{(1)}T_1^{(2)}~ \bigg| ~e_1^TH_2 = 0.
\end{align*}
\par
Next, we transform the condition $e_1^TH_2 = 0$ to the constraint on the distribution of $T_1^{(i)}$. Letting $t_1^2 = e_1^THH^Te_1$, then $e_1^TH_2 = 0$ is equivalent to $T_1^{(1)2} = e_1^THH^Te_1 = t_1^2$, which implies that
\begin{align*}
  e_1^THH^Te_2 \stackrel{(d)}{=} \quad T_1^{(1)}T_1^{(2)}~ \bigg|~ T_1^{(1)2} = e_1^THH^Te_1.
\end{align*}
\end{proof}

\begin{Prop}\label{prop:condinum}
  Assume the conditional number of $\Sigma$ is $cond(\Sigma)$ and $\Sigma_{ii} = 1$ for $i = 1,2, \cdots, p$, then we have
  \begin{align*}
    \lambda_{min}(\Sigma) \geq \frac{1}{cond(\Sigma)}\qquad\mbox{and}\qquad \lambda_{max}(\Sigma) \leq cond(\Sigma).
  \end{align*}
\end{Prop}

\begin{proof}
  Notice that $p = tr(\Sigma) = \sum_{i = 1}^p\lambda_i$. Therefore, we have
  \begin{align*}
    p/\lambda_{max}\geq \frac{p}{cond(\Sigma)}\quad\mbox{and}\quad p/\lambda_{min}(\Sigma) \leq p\cdot cond(\Sigma),
  \end{align*}
which completes the proof.
\end{proof}

We now turn to the proof of Lemma \ref{Lemma:5}.

\begin{proof}[\textbf{Proof of Lemma \ref{Lemma:5}}] Notice that to quantify $e_iHH^T\beta$ is essential to quantify the entries of $HH^T$. The diagonal terms are already studied in Lemma \ref{Lemma:4} as taking $v = e_i$ we have
  \begin{align}
    P\bigg(e_i^THH^Te_i<c_1'\frac{n^{1-\tau}}{p}\mbox{ or }e_i^THH^Te_i > c_2'\frac{n^{1+\tau}}{p}\bigg) < 4e^{-Cn}.\label{eq:t1}
  \end{align}
The remaining task is to quantify off diagonal terms. Without loss of generality, we prove the bound only for $e_1^THH^Te_2$, then the other off-diagonal terms should follow exactly the same argument. According to Proposition \ref{prop:decomp} with $q$ being 1, we can decompose $H = (T_1, H_2)$ with $T_1 = G(H_2)H_1$, where $H_2$ is a $p\times (n-1)$ matrix, $H_1$ is a $(p-n + 1)\times 1$ vector and $G(H_2)$ is a matrix such that $(G(H_2), H_2) \in \mathcal{O}(p)$.The invariant measure on the Stiefel manifold can be decomposed as
\begin{align*}
  [H] = [H_1][H_2]
\end{align*}
where $[H_1]$ and $[H_2]$ are Haar measures on $V_{1,n-p+1}, V_{n-1,p}$.  $H_1|H_2$ follows the Angular Central Gaussian distribution $ACG(\Sigma')$, where
\begin{align*}
  \Sigma' = G(H_2)^T\Sigma G(H_2).
\end{align*}

Let $H_1 = (h_1,h_2,\cdots,h_p)^T$ and let $x^T = (x_1,x_2,\cdots,x_{p-n+1})\sim N(0,\Sigma')$, then we have
\begin{align*}
  h_i\stackrel{(d)}{=} \frac{x_i}{\sqrt{x_1^2+\cdots+x_{p-n+1}^2}}.
\end{align*}
Notice that $T_1 = G(H_2)H_1$, a linear transformation on $H_1$. Defining $y = G(H_2)x$, we have
\begin{align}
  T_1^{(i)}\stackrel{(d)}{=} \frac{y_i}{\sqrt{y_1^2+\cdots+y_{p}^2}},\label{eq:y}
\end{align}
where $y\sim N(0, G(H)\Sigma'G(H)^T)$ is a degenerate Gaussian distribution. This degenerate distribution contains an interesting form. Letting $z\sim N(0, \Sigma)$, we know $y$ can be expressed as $y = G(H)G(H)^Tz$. Write $G(H_2)^T$ as $[g_1,g_2]$ where $g_1$ is a $(p-n+1)\times 1$ vector and $g_2$ is a $(p-n+1)\times (p-1)$ matrix, then we have
\begin{align*}
  G(H_2)G(H_2)^T =
  \begin{pmatrix}
    g_1^Tg_1 & g_1^Tg_2\\
    g_2^Tg_1 & g_2^Tg_2
  \end{pmatrix}
.
\end{align*}
We can also write $H_2^T = [0_{n-1,1}, h_2]$ where $h_2$ is a $(n-1)\times (p-1)$ matrix, and using the orthogonality, i.e., $[H_2~G(H_2)][H_2~G(H_2)]^T = I_p$, we have
\begin{align*}
  g_1^Tg_1 = 1, ~g_1^Tg_2 = 0_{1,p-1}\quad \mbox{and}\quad g_2^Tg_2 = I_{p-1} - h_2h_2^T.
\end{align*}
Because $h_2$ is a set of orthogonal basis in the $p-1$ dimensional space, $g_2^Tg_2$ is therefore an orthogonal projection onto the space $\{h_2\}^{\perp} $ and $g_2^Tg_2 = AA^T$ where $A = g_2^T(g_2g_2^T)^{-1/2}$ is a $(p-1)\times (p-n)$ orientation matrix on $\{h_2\}^{\perp}$. Together, we have
\begin{align*}
  y =
  \begin{pmatrix}
    1 & 0\\
    0 & AA^T
  \end{pmatrix}
z.
\end{align*}
This relationship allows us to marginalize $y_1$ out with $y$ following a degenerate Gaussian distribution.
\par
Now according to Proposition \ref{prop:offdiag} and assuming $t_1^2 = e_1^THH^Te_1$, we have
\begin{align*}
  e_1^THH^Te_2 \stackrel{(d)}{=} \quad T_1^{(1)}T_1^{(2)}~ \bigg|~ T_1^{(1)2} = t_1^2.
\end{align*}

Because the magnitude of $t_1$ has been obtained in \eqref{eq:t1}, we can now condition on the value of $T_1^{(1)}$ to obtain the bound on $T_1^{(2)}$. From $T_1^{(1)2} = t_1^2$, we have
\begin{align}
  (1 - t_1^2) y_1^2 = t_1^2(y_2^2 + y_3^2 +\cdots + y_{p}^2).\label{eq:cons}
\end{align}
Notice this constraint is imposed on the norm of $\tilde y = (y_2,~y_3,\cdots,y_{p})$ and is thus independent of $(y_2/\|\tilde y\|,\cdots, y_{p}/\|\tilde y\|)$. Equation \eqref{eq:cons} also implies that 
\begin{align}
(1-t_1^2)(y_1^2 + y_2^2 + \cdots + y_p^2) = y_2^2 + y_3^2 + \cdots + y_p^2.\label{eq:norm}
\end{align}
Therefore, combining \eqref{eq:y} with \eqref{eq:cons}, \eqref{eq:norm} and integrating $y_1$ out, we have
\begin{align*}
  T_1^{(i)}~|~T_1^{(1)}=t_1 &~\stackrel{(d)}{=} \frac{\sqrt{1-t_1^2}y_i}{\sqrt{y_2^2+\cdots+y_{p}^2}},\qquad i = 2,3,\cdots,p,
\end{align*}
where $(y_2,y_3,\cdots,y_{p})\sim N(0, AA^T \Sigma_{22} AA^T)$ with $\Sigma_{22}$ being the covariance matrix of $z_2,\cdots, z_p$. 
\par
To bound the numerator, we use the classical tail bound on the normal distribution as for any $t>0$, ($\sigma_i = \sqrt{var(y_i)}\leq \sqrt{\lambda_{max}(AA^T\Sigma_{22}AA^T)}\leq \lambda_{max}(\Sigma)^{1/2}$),
\begin{align}
  P(|y_i|>t\sigma_i) \leq 2e^{-t^2/2}.\label{eq:normal2}
\end{align}
For any $\alpha > 0$ choosing $t = \sqrt{C}\frac{n^{\frac{1}{2}-\alpha}}{\sqrt{\log n}}$ we have,
\begin{align*}
  P(|y_i|>\frac{\sqrt{C\lambda_{max}(\Sigma)}}{\sqrt{\log n}}n^{\frac{1}{2} - \alpha}) \leq 2\exp\bigg\{\frac{-Cn^{1-2\alpha}}{2\log n}\bigg\}.
\end{align*}

For the denominator, letting $\tilde z \sim N(0, I_{p-1})$, we have
\begin{align*}
  \tilde y = AA^T\Sigma_{22}^{1/2}\tilde z
  \quad\mbox{and}\quad \tilde y^T\tilde y = \tilde z^T \Sigma_{22}^{1/2}AA^T\Sigma_{22}^{1/2} \tilde z \stackrel{(d)}{=} \sum_{i = 1}^{p-n}\lambda_i \mathcal{X}^2_i(1),
\end{align*}
where $\mathcal{X}^2_i(1)$ are iid chi-square random variables and $\lambda_i$
are non-zero eigenvalues of matrix
$\Sigma_{22}^{1/2}AA^T\Sigma_{22}^{1/2}$. Here $\lambda_i$'s are naturally upper bounded by $\lambda_{max}(\Sigma)$. To give a lower bound, notice that $\Sigma_{22}^{1/2}AA^T\Sigma_{22}^{1/2}$ and $A\Sigma_{22}A^T$ possess the same set of non-zero eigenvalues, thus
\begin{align*}
  \min_i\lambda_i \geq \lambda_{min}(A\Sigma_{22}A^T) \geq \lambda_{min}(\Sigma).
\end{align*}
Therefore, 
\begin{align*}
  \lambda_{min}(\Sigma)\frac{\sum_{i=1}^{p-n}\mathcal{X}^2_i(1)}{p-n}\leq \frac{\tilde y^T\tilde y}{p-n}\leq  \lambda_{max}(\Sigma)\frac{\sum_{i=1}^{p-n}\mathcal{X}^2_i(1)}{p-n}.
\end{align*}
The quantity $\frac{\sum_{i=1}^{p-n}\mathcal{X}^2_i(1)}{p-n}$ can be bounded
by Proposition \ref{Prop:1}. Combining with Proposition \ref{prop:condinum}, we have for any $C>0$, there exists some $c_3'>0$ such that
\begin{align*}
  P\bigg(\tilde y^T\tilde y/(p-n) < \frac{c_3'}{\lambda_{min}(\Sigma)}\bigg)\leq e^{-C(p -n)}.
\end{align*}
Therefore, noticing that $cond(\Sigma) = \lambda_{max}(\Sigma)/\lambda_{min}(\Sigma) \leq c_4 n^\tau$, $T_1^{(2)}$ can be bounded as
\begin{align*}
  P\bigg(|T_1^{(2)}|> \frac{\sqrt{1 - t_1^2}\sqrt{Cc_4} n^{\frac{1}{2} + \frac{\tau}{2} - \alpha}}{\sqrt{c_3'}\sqrt{p-n}\sqrt{\log n}}~\big|T_1^{(1)} = t_1\bigg)\leq e^{-C(p-n)} + 2\exp\bigg\{\frac{-Cn^{1-2\alpha}}{2\log n}\bigg\}.
\end{align*}
Using the results from \eqref{eq:t1}, we have
\begin{align*}
  P\bigg(t_1^2>c_2'\frac{n^{1+\tau}}{p}\bigg)\leq 2e^{-Cn}.\quad\mbox{and}\quad
  P\bigg(t_1^2<c_1'\frac{n^{1-\tau}}{p}\bigg)\leq 2e^{-Cn}.
\end{align*}
Consequently, defining $M = \frac{\sqrt{Cc_4}}{\sqrt{c_3'(c_0-1)}}$ we have 
\begin{align*}
  P\bigg(|e_1^TH&H^Te_2|>\frac{M}{\sqrt{\log n}}\frac{n^{1 + \tau - \alpha}}{p} \bigg) = P\bigg(|T_1^{(1)}T_1^{(2)}|> \frac{M}{\sqrt{\log n}}\frac{n^{1 + \tau - \alpha}}{p}~\big| T_1^{(1)} = t_1\bigg) \\
&\leq P\bigg(T_1^{(1)2}>c_2'\frac{n^{1+\tau}}{p}~|T_1^{(1)} = t_1\bigg) + P\bigg(|T_1^{(2)}|> \frac{\sqrt{Cc_4} n^{\frac{1}{2} + \frac{\tau}{2} - \alpha}}{\sqrt{c_3'(c_0-1)}\sqrt{n}\sqrt{\log n}}~\big|T_1^{(1)} = t_1\bigg)\\
&\leq e^{-C(c_0-1)n} + 4e^{-Cn} + 2\exp\bigg\{\frac{-Cn^{1-2\alpha}}{2\log n}\bigg\}\\
& = O\bigg\{\exp\bigg(\frac{-Cn^{1-2\alpha}}{2\log n}\bigg)\bigg\}.
\end{align*}
This result provides an upper bound on off diagonal terms. Using this result, we have for any $i\not\in S$
\begin{align}
\notag  |e_i^THH^T\beta| &= \bigg|\sum_{j\in S}e_i^THH^Te_j\beta_j\bigg|\leq \sum_{j\in S}|e_i^THH^Te_j||\beta_j|\\
&\leq \sqrt{\sum_{j\in S}|e_i^THH^Te_j|^2}\cdot \|\beta\|_2
\leq \frac{\sqrt{c'c_4}M}{\sqrt{\log n}}\frac{n^{1+3\tau/2 + \nu/2 - \alpha}}{p}, \label{eq:off-alpha}
\end{align}
with probability at least $1 - O\bigg\{n^{\nu}\exp\bigg(\frac{-Cn^{1-2\alpha}}{2\log n}\bigg)\bigg\}$. The last inequality is due to Assumption 3 that $var(Y) = O(1)$, which implies
\begin{equation}
c_4^{-1}\|\beta\|^2n^{-\tau}\leq \|\beta\|^2\lambda_{min}(\Sigma) \leq \beta^T\Sigma\beta = var(Y) - \sigma^2 \leq c' \label{eq:con3}
\end{equation}
for some constant $c'$. Taking $\alpha = (5/2)\tau + \kappa + \nu/2$ in \eqref{eq:off-alpha}, we have
\begin{align}
  P\bigg(|e_i^THH^T\beta| > \frac{\tilde c}{\sqrt{\log n}}\frac{n^{1-\tau-\kappa}}{p}\bigg) = O\big\{\exp\bigg(\frac{-C'n^{1-5\tau - 2\kappa - \nu}}{2\log n}\bigg)\bigg\},
\end{align}
where $\tilde c = \sqrt{c'c_4}M$ and $C'$ is any constant that is less than $C$ ($C$ is also an arbitrary constant).
Next, for $i\in S$, combining \eqref{eq:off-alpha} with \eqref{eq:t1} and using the same value of $\alpha$ yields
\begin{align*}
  |e_i^THH^T\beta| &\geq |e_i^THH^Te_i||\beta_i| - \sum_{j\in S}|e_i^THH^Te_j||\beta_j|\\
&\geq c_1'c_2\frac{n^{1-\tau-\kappa}}{p} - \frac{\tilde c}{\sqrt{\log n}}\frac{n^{1 - \tau - \kappa}}{p}\\
&\geq \frac{c_1'c_2}{2}\frac{n^{1-\tau-\kappa}}{p},
\end{align*}
with probability at least $1 - 2e^{-Cn} - O\bigg\{n^{\nu}\exp\bigg(\frac{-Cn^{1-5\tau-2\kappa - \nu}}{2\log n}\bigg)\bigg\}$. Letting $c = c_1'c_2/2$ we have
\begin{align}
  P\bigg(|e_i^THH^T\beta| < c\frac{n^{1-\tau-\kappa}}{p}\bigg) = O\big\{\exp\bigg(\frac{-C'n^{1-5\tau - 2\kappa - \nu}}{2\log n}\bigg)\bigg\},
\end{align}
which completes the proof.

\end{proof}

\begin{proof}[{\bf Proof of Lemma \ref{Lemma:6}}]
 Recall the random variable $\eta_i = e_i^T\eta = e_i^TX^T(XX^T)^{-1}\epsilon$. If we define 
\[
a = e_i^TX^T(XX^T)^{-1}/\|e_i^TX^T(XX^T)^{-1}\|_2,
\]
then $a$ is independent of $\epsilon$ and 
\[
\eta_i \stackrel{(d)}{=} \|e_i^TX^T(XX^T)^{-1}\|_2\cdot \sigma w,
\]
where $w$ is a standardized random variable such that $w = a^T\epsilon/\sigma$. Now for the norm  we have
\begin{align}
\notag \|e_i^TX^T(XX^T)^{-1}\|_2^2 &=  e_i^TX^T(XX^T)^{-2}Xe_i = e_i^TX^T(XX^T)^{-1/2}(XX^T)^{-1}(XX^T)^{-1/2}Xe_i\\
\notag &\leq \lambda_{max}((XX^T)^{-1})\|(XX^T)^{-1/2}Xe_i\|^2\leq \lambda_{max}((XX^T)^{-1}) e_i^THH^Te_i\\
&= \lambda_{max}((Z\Sigma Z^T)^{-1}) e_i^THH^Te_i. \label{eq:var_eta}
\end{align}
The first term follows that
\begin{align}
\notag \lambda_{max}((Z\Sigma Z^T)^{-1}) &= (\lambda_{min}(Z\Sigma Z^T))^{-1}\leq \lambda_{min}(ZZ^T)^{-1}\lambda_{min}(\Sigma)^{-1} = p^{-1}\lambda_{min}(p^{-1}ZZ^T)^{-1}\lambda_{min}(\Sigma)^{-1}\\
&<\frac{c_4n^\tau}{p}\lambda_{min}(p^{-1}ZZ^T)^{-1},\label{eq:2}
\end{align}
where the last step is {\color{black}due to equation \eqref{eq:lambdamin}}.
According to A2, for some $C_1>0$ and $c_1>1$, we have
\[
P\bigg(\lambda_{max}(p^{-1}ZZ^T)>c_1\quad\mbox{or}\quad\lambda_{min}(p^{-1}ZZ^T)<c_1^{-1}\bigg)<\exp\big(-C_1n\big),
\]
which together with \eqref{eq:2} ensures that
\begin{align}
\notag P\bigg(\lambda_{max}\big((Z\Sigma Z^T)^{-1}\big)>\frac{c_1c_4n^{\tau}}{p}\bigg)&<P\bigg(\frac{c_4n^\tau}{p}(\lambda_{min}\big(p^{-1}ZZ^T)\big)^{-1}>\frac{c_1c_4n^{\tau}}{p}\bigg)\\
&= P\bigg(\lambda_{min}(p^{-1}ZZ^T)<c_1^{-1}\bigg)<e^{-C_1n}. \label{eq:3}
\end{align}
Combining Lemma \ref{Lemma:4} and \eqref{eq:3} entails that for the same $C_1>0$,
\begin{equation}
P\bigg(\|e_i^TX^T(XX^T)^{-1}\|_2^2> c_1c_2'c_4\frac{n^{1+2\tau}}{p^2}\bigg)< 3\exp\big(-C_1n\big). \label{eq:var}
\end{equation}
For $w$, according to the $q$-exponential tail assumption, we have
\begin{align*}
  P(|\sum_{i=1}^n a_i\epsilon_i/\sigma |> t)\leq \exp(1-q(t)).
\end{align*}
By choosing $t= \frac{\sqrt C_1n^{1/2-2\tau-\kappa}}{\sqrt{\log n}}$, we have
\[
P\bigg(|w|>  \frac{\sqrt{C_1}n^{1/2-2\tau-\kappa}}{\sqrt{\log n}}\bigg)< \exp\bigg\{1-q\bigg(\frac{\sqrt{C_1}n^{1/2-2\tau-\kappa}}{\sqrt{\log n}}\bigg)\bigg\}.
\]
Combining it with \eqref{eq:var}, taking the union bound, we have
\begin{equation}
P\bigg(|\eta_i|>\frac{\sigma\sqrt{C_1c_1c_2'c_4}}{\sqrt{\log n}}\frac{n^{1-\kappa-\tau}}{p}\bigg)<\exp\bigg\{1-q\bigg(\frac{\sqrt{C_1}n^{1/2-2\tau-\kappa}}{\sqrt{\log n}}\bigg)\bigg\}+ 3\exp\big(-C_1n\big).\label{eq:epsilon1}
\end{equation}
The proof is completed.
\end{proof}

\subsection*{Proof of the theorems}
By now we have all the technical results needed to prove the main theorems. The proof
of Theorem 1 follows the basic scheme of \cite{Fan:Lv:2008} but with a
modification of their step 2 by using our Lemma \ref{Lemma:5}}. The proof of Theorem 2 is a direct application of Lemma \ref{Lemma:5} and step 4 in the proof of Theorem 1. Theorem 3 mainly use the properties of Taylor expansion on matrix elements.

\par
\begin{proof}[\textbf{Proof of Theorem 1}] 
Applying Lemma \ref{Lemma:5} and Lemma \ref{Lemma:6} to all $i\in S$, we have
\begin{align}
P\bigg(\min_{i\in S}|\xi_i|<c\frac{n^{1-\tau-\kappa}}{p}\bigg) = O\bigg\{s\cdot \exp\bigg(\frac{-Cn^{1-5\tau-2\kappa-\nu}}{2\log n}\bigg)\bigg\}\label{eq:index.low1}
\end{align}
and
\begin{align}
  P\bigg(\max_{i\in S}|\eta_i|>\frac{\sigma\sqrt{C_1c_1c_2'c_4}}{\sqrt{\log n}}\frac{n^{1-\tau-\kappa}}{p}\bigg) = s\cdot \exp\bigg\{1-q\bigg(\frac{\sqrt C_1 n^{1/2-2\tau-\kappa}}{\sqrt{\log n}}\bigg)\bigg\} + 3s\cdot \exp\big(-C_1n\big).
\end{align}
Because $s = c_3n^\nu$ with $\nu<1$, taking $C = 2C_1$, \eqref{eq:index.low1} can be updated as
\begin{align}
  P\bigg(\min_{i\in S}|\xi_i|<c\frac{n^{1-\tau-\kappa}}{p}\bigg) = O\bigg\{\exp\bigg(\frac{-C_1n^{1-5\tau-2\kappa-\nu}}{2\log n}\bigg)\bigg\}.\label{eq:index.low2}
\end{align}
Therefore, if we choose $\gamma_n$ such that
\begin{align}
\frac{\tilde c+\sigma\sqrt{C_1c_1c_2'c_4}}{\sqrt{\log n}}\frac{n^{1-\tau-\kappa}}{p}<\gamma_n<\frac{c}{2}\frac{n^{1-\tau-\kappa}}{p},  
\end{align}
or in an asymptotic form satisfying
\begin{align}
  \frac{p\gamma_n}{n^{1-\tau-\kappa}}\rightarrow 0\quad\mbox{ and }\quad\frac{p\gamma_n\sqrt{\log n}}{n^{1-\tau-\kappa}}\rightarrow\infty,
\end{align}
then we have
\begin{align}
\notag  P\bigg(\min_{i\in S}|\hat\beta_i|<\gamma_n\bigg) &= P\bigg(\min_{i\in S}|\xi_i+\eta_i|<\gamma_n\bigg)\\
\notag &\leq P\bigg(\min_{i\in S}|\xi_i|<c\frac{n^{1-\tau-\kappa}}{p}\bigg)+P\bigg(\max_{i\in S}|\eta_i|>\frac{\sigma\sqrt{C_1c_1c_2'c_4}}{\sqrt{\log n}}\frac{n^{1-\tau-\kappa}}{p}\bigg)\\
\notag &= O\bigg\{\exp\bigg(\frac{-C_1n^{1-5\tau-2\kappa-\nu}}{2\log n}\bigg)\bigg\}+s\cdot \exp\bigg\{1-q\bigg(\frac{\sqrt C_1 n^{1/2-2\tau-\kappa}}{\sqrt{\log n}}\bigg)\bigg\}.
\end{align}
This completes the proof of Theorem 1.
\end{proof}

\par
\begin{proof}[\textbf{Proof of Theorem 2}]
According to Lemma \ref{Lemma:5}, for any $i\not\in S$ and any $C>0$, there exists a $\tilde c>0$ such that
\[
P\bigg(|e_iHH^T\beta|>\frac{\tilde c}{\sqrt{\log n}}\frac{n^{1-\tau-\kappa}}{p}\bigg)\leq O\bigg\{\exp\bigg(\frac{-Cn^{1-5\tau-2\kappa-\nu}}{2\log n}\bigg)\bigg\}.
\]
Now with Bonferroni's inequality, we have
\begin{align}
\notag P\bigg(\max_{i\not\in S}|\xi_i|>\frac{\tilde c}{\sqrt{\log n}}\frac{n^{1-\tau-\kappa}}{p}\bigg) &= P\bigg(\max_{i\not\in S}|e_iHH^T\beta|>\frac{\tilde c}{\sqrt{\log n}}\frac{n^{1-\tau-\kappa}}{p}\bigg)\\
&<O\bigg\{p\cdot \exp\bigg(\frac{-Cn^{1-5\tau-2\kappa-\nu}}{2\log n}\bigg)\bigg\}.
\end{align}
Also, applying Bonferroni's inequality to \eqref{eq:epsilon1} in the proof of Lemma \ref{Lemma:6} gives
\begin{equation*}
P\bigg(\max_ i |\eta_i|>\frac{\sigma\sqrt{C_1c_1c_2'c_4}}{\sqrt{\log n}}\frac{n^{1-\kappa-\tau}}{p}\bigg)<p\cdot \exp\bigg\{1-q\bigg(\frac{\sqrt C_1 n^{1/2-2\tau-\kappa}}{\sqrt{\log n}}\bigg)\bigg\} + 3p\cdot \exp\big(-C_1n\big).
\end{equation*}
Now recall that
\begin{align}
  \log p = o\bigg( \min\bigg\{ \frac{n^{1-2\kappa-5\tau}}{2\log n}, q\bigg(\frac{\sqrt C_1 n^{1/2-2\tau - \kappa}}{\sqrt{\log n}}\bigg)\bigg\}\bigg),
\end{align}
we have for the same $C_1$ specified in A2 (with the corresponding c),
\begin{align}
P\bigg(\max_{i\not\in S}|\xi_i|>\frac{\tilde c}{\sqrt{\log n}}\frac{n^{1-\tau-\kappa}}{p}\bigg)&< O\bigg\{\exp\bigg(-C_1\frac{n^{1-5\tau-2\kappa-\nu}}{2\log n}\bigg)\bigg\},\label{eq:24}\\
P\bigg(\max_ i |\eta_i|>\frac{\sigma\sqrt{C_1c_1c_2'c_4}}{\sqrt{\log n}}\frac{n^{1-\kappa-\tau}}{p}\bigg)&<O\bigg\{\exp\bigg(1-\frac{1}{2}q\bigg(\frac{\sqrt C_1 n^{1/2-2\tau-\kappa}}{\sqrt{\log n}}\bigg)\bigg)+\exp\bigg(-\frac{C_1}{2}n\bigg)\bigg\}.
\end{align}
Now if $\gamma_n$ is chosen as the same as in Theorem 1, we have
\begin{align*}
P\bigg(\max_{i\not\in S}|\hat\beta_i|>\gamma_n\bigg)<O\bigg\{\exp\bigg(-C_1\frac{n^{1-5\tau-2\kappa-\nu}}{2\log n}\bigg)+ \exp\bigg(1-\frac{1}{2}q\bigg(\frac{\sqrt C_1 n^{1/2-2\tau-\kappa}}{\sqrt{\log n}}\bigg)\bigg)\bigg\}.
\end{align*}
Therefore, combining the above result with Theorem 1 and noticing that $s<p$, we have
\begin{equation*}
P\bigg(\min_{i\in S}|\hat\beta_i|>\gamma_n>\max_{i\not\in S}|\hat\beta_i|\bigg) = 1 - O\bigg\{\exp\bigg(-C_1\frac{n^{1-5\tau-2\kappa-\nu}}{2\log n}\bigg)+ \exp\bigg(1-\frac{1}{2}q\bigg(\frac{\sqrt C_1 n^{1/2-2\tau-\kappa}}{\sqrt{\log n}}\bigg)\bigg)\bigg\}.
\end{equation*}
Obviously, if we choose a submodel with size $d$ that $d\geq s$ we will have
\begin{align}
  P(\mathcal{M}_S\subset \mathcal{M}_d) = 1 - O\bigg\{\exp\bigg(-C_1\frac{n^{1-5\tau-2\kappa-\nu}}{2\log n}\bigg)+ \exp\bigg(1-\frac{1}{2}q\bigg(\frac{\sqrt C_1 n^{1/2-2\tau-\kappa}}{\sqrt{\log n}}\bigg)\bigg)\bigg\},
\end{align}
which completes the proof of Theorem 2. 
\end{proof}
\par
\begin{proof}[\textbf{Proof of Corollary 1}]
  Replacing $q(t)$ by $C_0t^2/K^2$ for some $C_0>0$, the condition \eqref{eq:p} becomes
  \begin{align}
    \log p = o\bigg(\frac{n^{1-2\kappa-5\tau}}{\log n}\bigg),
  \end{align}
  and the result becomes
  \begin{align}
\notag  P\bigg(\min_{i\in S}|\hat\beta_i| > \gamma_n > \max_{i\not\in S}|\hat\beta_i|\bigg) &= 1 - O\bigg\{\exp\bigg(-C_1\frac{n^{1-2\kappa-5\tau-\nu}}{2\log n}\bigg) + \exp\bigg(-\frac{C_0C_1}{K^2}\frac{n^{1-2\kappa-4\tau}}{2\log n}\bigg)\bigg\}\\
& = 1 - O\bigg\{\exp\bigg(-C_1\frac{n^{1-2\kappa-5\tau-\nu}}{2\log n}\bigg)\bigg\}.
  \end{align}
The proof of Corollary 1 is completed.
\end{proof}
\par
\begin{proof}[\textbf{Proof of Theorem 3}]
It is intuitive to see that when the tuning parameter $r$ is sufficiently
small, the results in Theorem 2 should continue to hold. The issue here is to
find a better rate on $r$ to allow a more flexible choice for $r$. Following the ridge formula in Part A of the Supplementary Materials, the HOLP solution can be expressed as
\begin{align}
\hat\beta(r) = X^T(XX^T+rI_n)^{-1}X\beta+X^T(XX^T+rI_n)^{-1}\epsilon := \xi(r)+\eta(r).
\end{align}
\par
We look at $\xi(r)$ first. Using the notations in the very beginning of this section, we write
\begin{align*}
&X^T(XX^T+rI_n)^{-1}X = \Sigma^{1/2}UDV^T(VDU^T\Sigma UDV^T+rI_n)^{-1}VDU^T\Sigma^{1/2}\\
&= \Sigma^{1/2}U(U^T\Sigma U+rD^{-2})^{-1}U^T\Sigma^{1/2} = \Sigma^{1/2}U A^{-1}(I_n+rA^{-T}D^{-2}A^{-1})^{-1}A^{-T}U^T\Sigma^{1/2},
\end{align*}
{\color{black} where $A = (U^T\Sigma U)^{1/2}$, the square root of a positive definite symmetric matrix, i.e.,  $U^T\Sigma U = A^TA =  AA^T = A^2$.} In order to expand the
inverse matrix by Taylor expansion, we need to evaluate the largest eigenvalue
of the matrix $A^{-T}D^{-2}A^{-1}$, where
\begin{align}
\notag \lambda_{max}(A^{-T}D^{-2}A^{-1})\leq \lambda_{max}(D^{-2}){\color{black}\lambda_{max}\big((AA^T)^{-1}\big)} =\lambda_{min}(D^2)^{-1}\lambda_{min}(U^T\Sigma U)^{-1}.
\end{align}
According to A1 and A3, we have
\begin{align*}
P(p^{-1}\lambda_{min}(D^2)<c_1^{-1})<e^{-C_1n}
\end{align*}
for some $c_1>1$ and $C_1>0$ and
\begin{align*}
\lambda_{min}(U^T\Sigma U)\geq \lambda_{min}(\Sigma)\lambda_{min}(U^TU)\geq c_4^{-1}n^{-\tau}.
\end{align*}
Therefore, with probability greater than $1-e^{-C_1n}$, we have
\begin{align}
\lambda_{max}(A^{-T}D^{-2}A^{-1})\leq \frac{c_1c_4n^{\tau}}{p}, \label{eq:20}
\end{align}
meaning that when $r<pc_1^{-1}c_4^{-1}n^{-\tau}$, the norm of the matrix
$rA^{-T}D^{-2}A^{-1}$ is smaller than 1, and that the inverse of the matrix
can be expanded by the following Taylor series as
\begin{align*}
\Sigma^{1/2}U A^{-1}&(I_n+rA^{-T}D^{-2}A^{-1})^{-1}A^{-T}U^T\Sigma^{1/2} = \Sigma^{1/2}U A^{-1}(I_n+\sum_{k=1}^\infty r^k(A^{-T}D^{-2}A^{-1})^k)A^{-T}U^T\Sigma^{1/2}\\
&= HH^T + \sum_{k=1}^\infty r^k\Sigma^{1/2}U A^{-1}(A^{-T}D^{-2}A^{-1})^kA^{-T}U^T\Sigma^{1/2} = HH^T+M.
\end{align*}
The largest eigenvalue for each component of the infinite sum in the above formula can be bounded as
\begin{align*}
\lambda_{max}(\Sigma^{1/2}U &A^{-1}(A^{-T}D^{-2}A^{-1})^kA^{-T}U^T\Sigma^{1/2})\\
&\leq \lambda_{max}(\Sigma^{1/2}U A^{-1}A^{-T}U^T\Sigma^{1/2})\lambda_{max}\big((A^{-T}D^{-2}A^{-1})^k\big)\\
&= \lambda_{max}(HH^T)\lambda_{max}(A^{-T}D^{-2}A^{-1})^k\\
&\leq \lambda_{max}(A^{-T}D^{-2}A^{-1})^k,
\end{align*}
and so is their infinite sum as
\begin{align}
\lambda_{max}(M) \leq \sum_{k=1}^{\infty} r^k\lambda_{max}(A^{-T}D^{-2}A^{-1})^k\leq \frac{r\lambda_{max}(A^{-T}D^{-2}A^{-1})}{1-r\lambda_{max}(A^{-T}D^{-2}A^{-1})}.\label{eq:21}
\end{align}
The last step in the above formula requires $r\lambda_{max}(A^{-T}D^{-2}A^{-1})$ to be less than 1, and according to \eqref{eq:20}, it is true with probability greater than $1-e^{-C_1n}$.
Now with equation \eqref{eq:20} and \eqref{eq:21}, we have
\begin{align}
\notag P\bigg(\lambda_{max}(M)>\frac{c_1c_4rn^{\tau}}{p-c_1c_4rn^{\tau}}\bigg)&<P\bigg(\frac{r\lambda_{max}(A^{-T}D^{-2}A^{-1})}{1-r\lambda_{max}(A^{-T}D^{-2}A^{-1})}>\frac{c_1c_4rn^{\tau}}{p-c_1c_4rn^{\tau}}\bigg)\\
&= P\bigg(\lambda_{max}(A^{-T}D^{-2}A^{-1})> \frac{c_1c_4n^{\tau}}{p}\bigg) <e^{-C_1n}.\label{eq:M}
\end{align}
With the above equation, the following steps are straightforward. By choosing an appropriate rate of $r$, the entries of $M$ will be much smaller than the entries of $HH^T$, and the results established in Theorem 2 will remain valid. For any $i\in \{1,2,\cdots,p\}$, according to \eqref{eq:con3} we have
\begin{align*}
\max_{i\in\{1,\cdots,p\}}|e_i^TM\beta|^2 \leq \beta^T M^2\beta \leq \|\beta\|^2\lambda_{max}(M)^2\leq c'c_4n^\tau\lambda_{max}(M)^2.
\end{align*}
Hence if $r$ satisfies  $rn^{5/2\tau+\kappa-1}\rightarrow 0$ to ensure 
\begin{align*}
\frac{c_1c_4\sqrt{c'c_4}rn^{5/2\tau+\kappa-1}}{1-c_1c_4rn^\tau/p} = o(1),
\end{align*}
we can obtain an upper bound on $\max_i|e_i^TM\beta|$ following 
\eqref{eq:M} as
\begin{align*}
P\bigg(\max_{i\in\{1,\cdots,p\}}|e_i^TM\beta|>&\frac{n^{1-\tau-\kappa}}{p}\cdot\frac{c_1c_4\sqrt{c'c_4}rn^{5/2\tau+\kappa-1}}{1-c_1c_4rn^\tau/p}\bigg)\\
&<P\bigg(\sqrt{c'c_4n^{\tau}}\lambda_{max}(M)>\frac{c_1c_4\sqrt{c'c_4}rn^{\frac{3}{2}\tau}}{p-c_1c_4rn^{\tau}}\bigg)<e^{-C_1n}.
\end{align*}
Thus,
\begin{align*}
P\bigg(\max_{i\in\{1,\cdots,p\}}|e_i^TM\beta|>o(1)\frac{n^{1-\tau-\kappa}}{p}\bigg)<e^{-C_1n}.
\end{align*}
Recall the fact that $\xi_i(r) = \xi_i+e_iM\beta$. Combining the above equation and the results obtained in \eqref{eq:index.low2} and \eqref{eq:24}, similar properties for $\xi_i(r)$ can be established as 
\begin{align}
\notag P\bigg(\max_{i\not\in S}|\xi_i(r)|>o(1)\frac{n^{1-\tau-\kappa}}{p}\bigg)&< O\bigg\{\exp\bigg(-C_1\frac{n^{1-5\tau-2\kappa-\nu}}{2\log n}\bigg)\bigg\},\\
P\bigg(\min_{i\in S} |\xi_i(r)|<\frac{c}{2}\frac{n^{1-\kappa-\tau}}{p}\bigg)&<O\bigg\{\exp\bigg(-C_1\frac{n^{1-5\tau-2\kappa-\nu}}{2\log n}\bigg)\bigg\},\label{eq:xi_r}
\end{align}
where $c$ is some positive constant and $o(1)$ represents $\frac{\tilde c}{\sqrt{\log n}} + \frac{c_1c_4\sqrt{c'c_4}rn^{5/2\tau+\kappa-1}}{1-c_1c_4rn^\tau/p}$.
\par
Second, we look at $\eta(r)$. By a similar argument, the previous results on $\eta$ in Theorem 2 can be generalized. Because
\begin{align*}
\eta_i(r) = e_i^TX^T(XX^T+rI_n)^{-1}\epsilon,
\end{align*}
it follows
\begin{align*}
var(\eta_i(r)|X) &= \sigma^2e_i^TX^T(XX^T+rI_n)^{-2}Xe_i\\
&= \sigma^2e_i^TX^T(XX^T+rI_n)^{-1/2}(XX^T+rI_n)^{-1}(XX^T+rI_n)^{-1/2}Xe_i\\
&\leq \sigma^2\lambda_{max}((XX^T+rI_n)^{-1})\cdot e_i^TX^T(XX^T+rI_n)^{-1}Xe_i\\
& = \sigma^2(\lambda_{min}(XX^T+rI_n))^{-1}\cdot e_i^TX^T(XX^T+rI_n)^{-1}Xe_i.
\end{align*}
Using the same notation in the $\xi(r)$ part, we have
\begin{align*}
e_i^TX^T(XX^T+rI_n)^{-1}Xe_i = e_i^THH^Te_i+e_i^TMe_i \leq e_i^THH^Te_i+\lambda_{max}(M),
\end{align*}
and for $\lambda_{min}(XX^T+rI_n)$, it can be expressed as
\begin{align*}
\lambda_{min}(XX^T+rI_n) = r+\lambda_{min}(XX^T)\geq \lambda_{min}(XX^T).
\end{align*}
Therefore, the conditional variance of $\eta_i(r)$ can be reformulated as
\begin{align}
\notag var(\eta_i(r)|X) &\leq \sigma^2 \big(\lambda_{min}(XX^T)\big)^{-1}(e_i^THH^Te_i+\lambda_{max}(M)) \\
\notag &= \sigma^2\big(\lambda_{min}(XX^T)\big)^{-1} e_i^THH^Te_i+\sigma^2\big(\lambda_{min}(XX^T)\big)^{-1}\lambda_{max}(M) \\
&= \sigma^2\big(\lambda_{min}(Z\Sigma Z^T)\big)^{-1} e_i^THH^Te_i+\sigma^2\big(\lambda_{min}\big(Z\Sigma Z^T)\big)^{-1}\lambda_{max}(M). \label{eq:eta_r}
\end{align}
The first term in the above formula appears as $var(\eta_i|X)$ in \eqref{eq:var_eta} in the proof of Lemma \ref{Lemma:6}, while the second term $\sigma^2(\lambda_{min}(Z\Sigma Z^T))^{-1}\lambda_{max}(M)$ is introduced by the ridge parameter $r$. If we are able to show that this new conditional variance has the same bound as specified in equation \eqref{eq:var} (with a different constant), then a similar result of Lemma \ref{Lemma:6} can also be established for $\eta_i(r)$, i.e.,
\begin{align}
P\bigg(|\eta_i(r)|>\frac{\sigma\sqrt{2C_1c_1c_2'c_4}}{\sqrt{\log n}}\frac{n^{1-\kappa-\tau}}{p}\bigg)<\exp\bigg\{1-q\bigg(\frac{\sqrt{C_1}n^{1/2-2\tau-\kappa}}{\sqrt{\log n}}\bigg)\bigg\}+ 5\exp\big(-C_1n\big)\label{eq:eta_r1}.
\end{align}
and therefore,
\begin{align}
P\bigg(\max_{i}|\eta_i(r)|>\frac{\sigma\sqrt{2C_1c_1c_2'c_4}}{\sqrt{\log n}}\frac{n^{1-\kappa-\tau}}{p}\bigg)<O\bigg\{\exp\bigg(1-\frac{1}{2}q\bigg(\frac{\sqrt{C_1}n^{1/2-2\tau-\kappa}}{\sqrt{\log n}}\bigg)\bigg)+ \exp\bigg(-\frac{C_1}{2}n\bigg)\bigg\}\label{eq:eta_r2}.
\end{align}
\par
In fact, a similar equation as \eqref{eq:var} can be verified for this new variance directly from \eqref{eq:3} and \eqref{eq:M}. Since $\big(\lambda_{min}(Z\Sigma Z^T)\big)^{-1} = \lambda_{max}\big((Z\Sigma Z^T)^{-1}\big)$, by inequality \eqref{eq:3} and \eqref{eq:M} we have
\begin{align*}
P\bigg(\sigma^2\big(\lambda_{min}(Z\Sigma Z^T)\big)^{-1}\lambda_{max}(M)> \sigma^2\frac{c_1c_4n^\tau}{p}\cdot\frac{c_1c_4rn^\tau}{p-c_1c_4rn^\tau}\bigg)< 2e^{-C_1n}.
\end{align*}
Rearranging the lower bound in the probability gives
\begin{align*}
P\bigg(\sigma^2(\lambda_{min}\big(Z\Sigma Z^T)\big)^{-1}\lambda_{max}(M)> \frac{n^{1+2\tau}}{p^2}\cdot \frac{\sigma^2c_1^2c_4^2 rn^{-1}}{1-c_1c_4rn^\tau/p}\bigg)< 2e^{-C_1n}.
\end{align*}
If $r$ satisfies the condition stated in the theorem ensuring
\begin{align*}
\frac{\sigma^2c_1^2c_4^2 rn^{-1}}{1-c_1c_4rn^\tau/p} = o(1),
\end{align*}
then it holds that
\begin{align*}
 P\bigg(\sigma^2\big(\lambda_{min}(Z\Sigma Z^T)\big)^{-1}\lambda_{max}(M)>o(1)\frac{n^{1+2\tau}}{p^2}\bigg)<2e^{-C_1n},
\end{align*}
which combined with \eqref{eq:var} and \eqref{eq:eta_r} entails that
\begin{align*}
P\bigg(var(\eta_i(r)|X)>2c_1c_2'c_4\frac{n^{1+2\tau}}{p^2}\bigg)< 5e^{-C_1n},
\end{align*}
and thus proves equation \eqref{eq:eta_r2} (by following the argument in Lemma \ref{Lemma:6}).
\par
Finally, combining \eqref{eq:xi_r} and \eqref{eq:eta_r2} we have,
\begin{align*}
\notag P\bigg(\min_{i\in S}|\hat\beta_i(r)|< \frac{c}{4}\frac{n^{1-\tau-\kappa}}{p}\bigg)&<O\bigg\{\exp\bigg(-C_1\frac{n^{1-5\tau-2\kappa-\nu}}{2\log n}\bigg) + \exp\bigg(1-q\bigg(\frac{\sqrt{C_1}n^{1/2-2\tau-\kappa}}{\sqrt{\log n}}\bigg)\bigg)\bigg\}\\
P\bigg(\max_{i\not\in S}|\hat\beta_i(r)|> o(1)\frac{n^{1-\tau-\kappa}}{p}\bigg)&<O\bigg\{\exp\bigg(-C_1\frac{n^{1-5\tau-2\kappa-\nu}}{2\log n}\bigg) + \exp\bigg(1-q\bigg(\frac{\sqrt{C_1}n^{1/2-2\tau-\kappa}}{\sqrt{\log n}}\bigg)\bigg)\bigg\}.
\end{align*}
where $o(1)$ is used to denote $\frac{\tilde c+\sigma\sqrt{2C_1c_1c_2'c_4}}{\sqrt{\log n}} + \frac{c_1c_4\sqrt{c'c_4}rn^{5/2\tau+\kappa-1}}{1-c_1c_4rn^\tau/p}$, which is an infinitesimal.
Therefore, if we choose $\gamma_n$ such that
\begin{align}
 \bigg(\frac{\tilde c+\sigma\sqrt{2C_1c_1c_2'c_4}}{\sqrt{\log n}} + \frac{c_1c_4\sqrt{c'c_4}rn^{5/2\tau+\kappa-1}}{1-c_1c_4rn^\tau/p}\bigg)\frac{n^{1-\kappa-\tau}}{p} < \gamma_n < \frac{c}{4}\frac{n^{1-\kappa-\tau}}{p},
\end{align}
or in asymptotic form, 
\begin{align}
   \frac{\gamma_n p}{n^{1-\kappa-\tau}}\rightarrow 0\quad\mbox{and}\quad\frac{\gamma_n p\sqrt{\log n}}{n^{1-\kappa-\tau}}\rightarrow \infty\quad\mbox{and}\quad\frac{\gamma_n p}{rn^{\frac{3}{2}\tau}}\rightarrow\infty.
\end{align}
Then we can conclude similarly as Theorem \ref{thm:1} and \ref{thm:2} that
\begin{align*}
P\bigg(\min_{i\in S}|\hat\beta_i(r)|>&\gamma_n>\max_{i\not\in S}|\hat\beta_i(r)|\bigg) \\
&= 1 - O\bigg\{\exp\bigg(-C_1\frac{n^{1-5\tau-2\kappa-\nu}}{2\log n}\bigg)+ \exp\bigg(1-\frac{1}{2}q\bigg(\frac{\sqrt C_1 n^{1/2-2\tau-\kappa}}{\sqrt{\log n}}\bigg)\bigg)\bigg\}.
\end{align*}
Obviously, if we choose a submodel with size $d$ that $d\geq s$ we will have
\begin{align}
  P\bigg(\mathcal{M}_S\subset \mathcal{M}_d\bigg) = 1 - O\bigg\{\exp\bigg(-C_1\frac{n^{1-5\tau-2\kappa-\nu}}{2\log n}\bigg)+ \exp\bigg(1-\frac{1}{2}q\bigg(\frac{\sqrt C_1 n^{1/2-2\tau-\kappa}}{\sqrt{\log n}}\bigg)\bigg)\bigg\}.
\end{align}
The proof is now completed.
\end{proof}

\setcounter{table}{0}
\renewcommand{\thetable}{S.\arabic{table}}
\setcounter{figure}{0}
\renewcommand{\thefigure}{S.\arabic{figure}}

\section*{D: Additional simulation}
This section contains the plots for ridge-holp in Simulation study 2 in Section 4.2 as well as the detailed simulation results for $(p,n)=(1000,100)$ in Simulation Study 1 in Section 4.1 and Simulation Study 4 in Section 4.4. 

\begin{figure}[!htbp]
  \centering
  \includegraphics[width = 16cm]{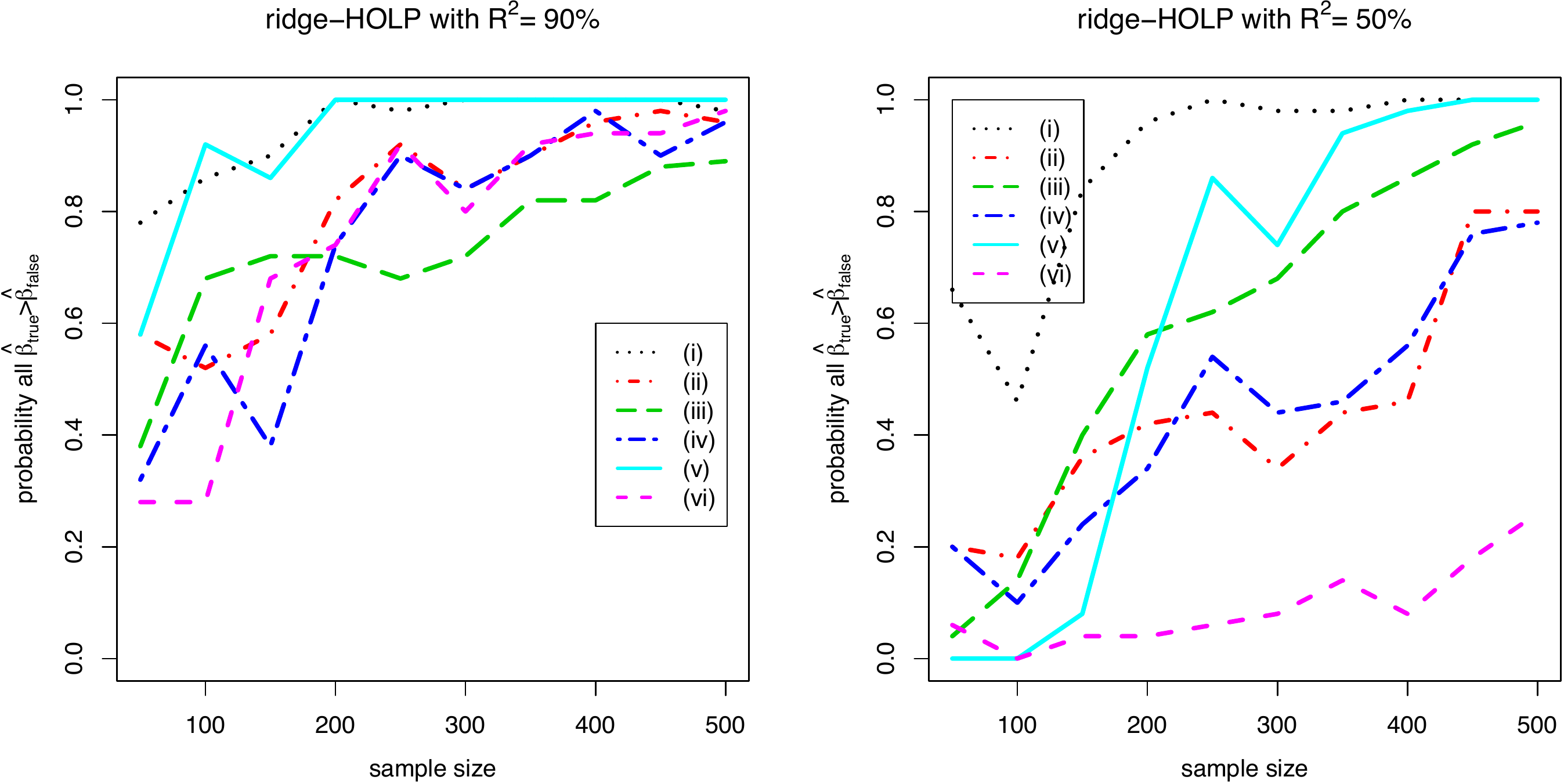}
  \caption{ridge-HOLP ($r=10$): $P(\min_{i\in S}|\hat\beta_i|>\max_{i\not\in
      S}|\hat\beta_i|)$ versus sample size $n$.}
  \label{fig:rholp}
\end{figure}

\begin{table}[!htbp]
\footnotesize
  \centering
  \caption{The probability to include the true model when $(p,n) = (1000,100)$ for Simulation Study 1 in Section 4.1}
  \begin{tabular}{l|l|l|cccccc}
\hline
 &Example &  & HOLP & SIS & RRCS & ISIS & FR & Tilting \\
\hline\hline
\multirow{14}{*}{$\mathcal{R}^2 = 50\%$}&
(i) Independent predictors &  & 0.685 & 0.690 & 0.615 & 0.270 & 0.370 & 0.340\\
\cline{2-9}
&\multirow{3}{*}{(ii) Compound symmetry}
                   & $\rho = 0.3$ & 0.195 & 0.135 & 0.195 & 0.050 & 0.005 & 0.000  \\
&                  & $\rho = 0.6$ &  0.020 & 0.010 & 0.040 & 0.005 & 0.000 & 0.000  \\
&                  & $\rho = 0.9$ &  0.000 & 0.000 & 0.000 & 0.000 & 0.000 & 0.010  \\
\cline{2-9}
&\multirow{3}{*}{(iii) Autoregressive}& $\rho = 0.3$ & 0.810 & 0.810 & 0.790 & 0.510 & 0.555 & 0.525\\
&                   & $\rho = 0.6$ & 0.970 & 0.985 & 0.970  & 0.560 & 0.390 & 0.355\\
&                  & $\rho = 0.9$ & 0.990 & 1.000 & 1.000 & 0.500 & 0.185 & 0.160\\
\cline{2-9}
&\multirow{3}{*}{(iv) Factor models}& $k = 2$ & 0.295 & 0.000 &0.000 & 0.045 & 0.135 & 0.105\\
&                      & $k = 10$&  0.060 & 0.000 &0.000 & 0.000 & 0.000 & 0.025\\
&                       & $k = 20$& 0.010 & 0.000 &0.000 & 0.000 & 0.000 & 0.000\\
\cline{2-9}
&\multirow{3}{*}{(v) Group structure} & $\delta^2 = 0.1$ & 0.935 & 0.970 & 0.950 & 0.000 & 0.000 & 0.000\\
&                      & $\delta^2 = 0.05$&  0.950 & 0.970 & 0.950 & 0.000 & 0.000 & 0.000\\
&                      & $\delta^2 = 0.01$&  0.960 & 0.980 & 0.970 & 0.000 & 0.000 & 0.000\\
\cline{2-9}
&(vi) Extreme correlation&                  & 0.305 & 0.000 &0.000 & 0.000 & 0.000 & 0.020 \\
\hline
\multirow{14}{*}{$\mathcal{R}^2 = 90\%$}&
(i) Independent predictors &  & 1.000 & 0.995 & 0.990 & 0.990 & 1.000 & 1.000\\
\cline{2-9}
&\multirow{3}{*}{(ii) Compound symmetry}
                   & $\rho = 0.3$ & 0.980 & 0.815 & 0.705 & 0.955 & 1.000 & 0.990 \\
&                  & $\rho = 0.6$ &  0.830 & 0.580 & 0.435  & 0.305 & 0.575 & 0.490 \\
&                  & $\rho = 0.9$ &  0.100 & 0.030 & 0.055 & 0.005 & 0.000 & 0.050 \\
\cline{2-9}
&\multirow{3}{*}{(iii) Autoregressive}& $\rho = 0.3$ & 0.990 & 0.965 & 0.945 & 1.000 & 1.000 & 1.000\\
&                   & $\rho = 0.6$ & 1.000 & 1.000 & 1.000 & 1.000 & 1.000 & 1.000 \\
&                  & $\rho = 0.9$ & 1.000 & 1.000 & 1.000 & 0.970 & 0.985 & 1.000 \\
\cline{2-9}
&\multirow{3}{*}{(iv) Factor models}& $k = 2$ & 0.940 & 0.015 & 0.000 &0.490 & 0.950 & 0.960\\
&                      & $k = 10$&  0.715 & 0.000 & 0.000 & 0.115 & 0.370 & 0.455\\
&                       & $k = 20$& 0.430 & 0.000 & 0.000 &0.015 & 0.105 & 0.225\\
\cline{2-9}
&\multirow{3}{*}{(v) Group structure} & $\delta^2 = 0.1$ & 1.000 & 1.000 & 1.000 & 0.000 & 0.000 & 0.000\\
&                      & $\delta^2 = 0.05$& 1.000 & 1.000 &1.000 & 0.000 & 0.000 & 0.000\\
&                      & $\delta^2 = 0.01$& 1.000 & 1.000 &1.000 & 0.000 & 0.000 & 0.000\\
\cline{2-9}
& (vi) Extreme correlation &                  & 0.905 & 0.000 &0.000 & 0.000 & 0.150 & 0.110\\
\hline
\end{tabular}\label{tab:1}
\end{table}

\begin{savenotes}
\begin{table}[!htbp]
\scriptsize
  \centering
  \caption{Model selection results when $(p, n) = (1000,100)$ for Simulation Study 4 in Section 4.4}\label{tab:iv1}
  \begin{tabular}{l|l|ccccccl}
\hline
    example & &\#FNs& \#FPs & Coverage(\%) & Exact(\%)  & Size & $||\hat\beta-\beta||_{2}$ & time (sec) \\
\hline\hline
\multirow{9}{*}{\specialcell{(i) Independent\\predictors\\ ~\\ $s = 5,
    ||\beta||_2 = 3.8$}}
                             & Lasso      & 0.05 & 0.78  & 95.0 & 44.5 &5.73  & 2.20 &0.09 \\
                             & SCAD       & 0.00 & 0.02  &100.0 & 98.0 &5.02  & 0.59 &1.37 \\
                             & ISIS-SCAD  & 0.01 & 0.05  & 99.0 & 94.0 & 5.04 & 0.59 &11.22\\
                             & SIS-SCAD   & 0.10 & 0.04  & 91.5 & 89.5 & 4.94 & 0.90 &0.19\\
                             & RRCS-SCAD  & 0.16 & 0.07  & 86.5 & 85.5 & 4.91 & 1.09 &0.54\\
                             & FR-Lasso   & 0.01 & 0.78  & 99.0 & 55.0 &5.62  & 1.75 &67.87 \\
                             & FR-SCAD    & 0.00 & 0.04  & 100.0& 96.0 &5.04  & 0.59 &68.02 \\
                             & HOLP-Lasso & 0.09 & 0.82  & 94.0 & 43.5 &5.72  &2.23  & 0.06\\
                             & HOLP-SCAD  & 0.06 & 0.04  & 95.5 & 94.0 &4.98  & 0.61 &0.20 \\
                             & HOLP-EBICS  & 0.18 & 0.02  & 83.5 & 83.0 &4.84  &0.93 &0.05 \\
                             & Tilting    & 0.00 & 0.05  & 100.0& 95.0 &5.05  & 0.59 & 294.7\\
\hline
\multirow{9}{*}{\specialcell{(ii) Compound\\symmetry\\~\\ $s = 5, ||\beta||_2 = 8.6$ }}    
                             & Lasso      & 2.64 & 2.38  &9.0   & 0.0 &4.74  & 9.85  &0.12\\
                             & SCAD       & 0.28 & 8.15 & 75.5 & 1.0 &12.97 & 8.33 & 5.64\\
                             & ISIS-SCAD  & 1.52 & 5.82  & 22.0 & 1.5 &9.30  & 8.68  & 20.15\\
                             & SIS-SCAD   & 1.59 & 4.58  & 24.0 & 5.5 &7.99  & 8.83 & 0.58\\
                             & RRCS-SCAD  & 1.79 & 4.78  & 18.5 & 5.0 &7.99  & 9.22 & 1.04\\
                             & FR-Lasso   & 0.93 & 5.43  &50.0  &1.0  &9.50  & 8.51 &92.81\\
                             & FR-SCAD    & 0.74 & 6.50  &56.0  & 1.0 &10.76 & 7.26 & 93.24\\
                             & HOLP-Lasso & 2.41 & 2.48  &12.0  &0.0  & 5.07 & 9.61  &0.10\\
                             & HOLP-SCAD  & 0.34 & 5.58  &72.5  & 3.0 &10.24 & 6.84  &0.52\\
                             & HOLP-EBICS  & 1.06 & 2.62  & 34.0 & 11.0& 6.56 & 7.16  & 0.19\\
                             & Tilting    & 3.07 & 5.07  & 20.0 & 0.0 & 7.00 & 9.82 & 238.2\\
\hline
\multirow{9}{*}{\specialcell{(iii) Autoregressive\\correlation\\~\\ $s = 3, ||\beta||_2 = 3.9$}}        
                             & Lasso      & 0.00 &1.12 &100.0   & 0.0  &4.12 &0.84 & 0.12\\
                             & SCAD       & 0.00 &0.03 &100.0   & 97.5 &3.66 &0.36 & 1.31\\
                             & ISIS-SCAD  & 0.00 &0.01 &100.0   & 99.0 &3.01 & 0.30 &14.27\\
                             & SIS-SCAD   & 0.00 &0.02 & 100.0  & 98.5 &3.02 & 0.30 &0.16\\
                             & RRCS-SCAD  & 0.00 &0.01 &100.0   & 99.0 &3.01 & 0.30 &0.66\\
                             & FR-Lasso   & 0.00 &1.12 & 100.0  & 0.0  &4.12 & 0.73 &96.01\\
                             & FR-SCAD    & 0.00 &0.04 & 100.0  & 96.5 &4.70 & 0.46 &96.19\\
                             & HOLP-Lasso & 0.00 &1.16 & 100.0  & 0.0  &4.16 & 0.83 & 0.06\\
                             & HOLP-SCAD  & 0.00 &0.01 & 100.0  & 99.0 &3.01 & 0.28 &0.15\\
                             & HOLP-EBICS  & 0.00 &0.00 & 100.0  & 100.0&3.00 & 0.28 &0.10\\
                             & Tilting    & 0.00 &0.01& 100.0  & 99.0 &3.01 &0.28 & 233.9\\
\hline
\multirow{9}{*}{\specialcell{(iv) Factor Models\\~\\$s = 5, ||\beta||_2 = 8.6$}}      
                             & Lasso      & 4.37 & 4.89  & 0.5  & 0.0 &5.52 & 11.20 &0.13\\
                             & SCAD       & 0.30 & 20.18 & 75.0 & 0.0 &24.88 &13.29 &3.23\\
                             & ISIS-SCAD  & 2.64 & 14.82 & 7.0  & 2.0 &17.19 & 15.50 &18.80\\
                             & SIS-SCAD   & 3.89 & 15.94 & 0.5  & 0.0 &17.05 & 16.58 & 0.66\\
                             & RRCS-SCAD  & 3.86 & 16.54 & 0.5  & 0.0 &17.68& 16.62 &1.03\\
                             & FR-Lasso   & 4.77 & 3.90  & 0.5  & 0.0 &6.13 & 11.92 &95.86\\
                             & FR-SCAD    & 2.87 & 12.08 & 16.0 & 0.0 &14.21& 15.80 &96.32\\
                             & HOLP-Lasso & 3.83 & 4.41  & 0.5  & 0.0 &5.58 & 11.20 &0.07\\
                             & HOLP-SCAD  & 0.81 & 11.22 & 65.0 & 2.5 &15.41 & 10.72 &0.67\\
                             & HOLP-EBICS  & 1.45 & 7.29  & 29.0 & 6.0 &10.84&11.37 &0.11\\
                             & Tilting    & 2.98 & 3.15  & 14.3 & 2.0 &5.17 & 12.02 & 165.8\\
\hline
\multirow{9}{*}{\specialcell{(v) Group\\structure\\~\\ $s = 5, ||\beta||_2 = 19.4$}} 
                             & Lasso          & 9.34  & 0.12  &0.0 &0.0& 5.77 & 14.08 & 0.13\\
                             & SCAD           & 11.94 & 63.76 &0.0 &0.0& 66.82 & 26.64 &3.74\\
                             & ISIS-SCAD      & 11.97 & 19.20 &0.0 &0.0&22.24 & 22.90 & 21.08\\
                             & SIS-SCAD       & 11.93 & 18.33 &0.0 &0.0&21.39 & 22.57 & 0.54\\
                             & RRCS-SCAD      & 11.93 & 18.03 &0.0 &0.0&21.10 & 22.65 & 0.94\\
                             & FR-Lasso       & 11.51 & 0.73  &0.0 &0.0& 4.22 & 18.52 & 96.64\\
                             & FR-SCAD        & 11.96 & 19.48 &0.0 &0.0& 23.84& 24.84 & 97.02\\
                             & HOLP-Lasso & 9.29  & 0.12  &0.0 &0.0& 5.83 & 14.07 &0.11\\
                             & HOLP-SCAD      & 11.93 & 18.32 &0.0 &0.0& 21.38 & 22.52 &0.42\\
                             & HOLP-EBICS      & 11.95 & 1.03  &0.0 &0.0 &4.08 & 22.32 &0.18\\
                             & Tilting        & 11.75 & 0.37  &0.0 &0.0& 3.62 & 22.95 &175.5\\
\hline
\multirow{9}{*}{\specialcell{(vi) Extreme\\correlation\\~\\ $s = 5, ||\beta||_2 = 8.6$}} 
                             & Lasso      & 2.35 & 8.20  &8.0  &0.0 & 10.84 & 10.14 &0.12\\
                             & SCAD       & 0.03 & 0.10  &97.5 &92.0& 5.07 & 1.93 & 3.28\\
                             & ISIS-SCAD  & 4.84 & 3.78  &0.0  &0.0 & 3.94 & 13.80 & 20.00\\
                             & SIS-SCAD   & 4.98 & 2.07  &0.0  &0.0 & 2.09 & 12.35 & 0.89\\
                             & RRCS-SCAD  & 4.98 & 2.06  &0.0  &0.0 & 2.08 & 12.35 & 1.38\\
                             & FR-Lasso   & 2.89 & 6.41  &1.0  &0.0 & 8.52 & 11.28 & 87.24\\
                             & FR-SCAD    & 3.08 & 3.15  &0.5  &0.5 & 5.07 & 12.41 & 87.60\\
                             & HOLP-Lasso & 2.26 & 8.16  &8.5  &0.0 & 10.90 & 10.13 &0.09\\
                             & HOLP-SCAD  & 0.03 & 0.08  &97.5 &93.5& 5.05 & 1.46 & 0.52\\
                             & HOLP-EBICS  & 0.70 & 0.70  &46.0 &46.0& 5.00 & 5.91 &0.16\\
                             & Tilting    & 4.47 & 3.17  &0.0  &0.0 & 3.70 & 12.54 & 208.8\\
\hline
  \end{tabular}
\end{table}
\end{savenotes}

\section*{E: An analysis of the commonly selected genes}
We summarize the commonly selected genes in Table 
\ref{tab:7}. The listed genes are all selected by at least two
methods. In particular, gene BE107075 is chosen by all methods other than
tilting. \cite{Breheny:Huang:2013} reported that this gene is also selected
via group Lasso and group SCAD, {\color{black}and we find that by fitting a
  cubic curve, it can explain more
than 65\% of the variance of TRIM32. Interestingly, tilting selects a
completely different set of genes, and even the submodel after screening is
thoroughly different from other screening methods. This result may be
explained by the strong correlations among genes, as the largest absolute
correlation is around 0.99 {\color{black} and the median is 0.62}. 
\begin{table}[!ht]
\begin{center}
\caption{Commonly selected genes for different methods}\label{tab:7}
\begin{tabular}{l|cccc}
  \hline\hline
  Probe ID & 1376747 &1381902 &1390539 &1382673\\
  \hline
  Gene name &BE107075&Zfp292&BF285569&BE115812\\
  \hline
  Lasso& yes & yes & yes &  \\
  SCAD& yes & yes &yes &  \\
  \hline
  ISIS-SCAD&yes & & & \\
  SIS-SCAD&yes&yes& & \\
  RRCS-SCAD &yes&yes& & \\
  \hline
  FR-Lasso &yes & yes & & \\
  FR-SCAD  &yes & yes & & yes \\
  \hline
  HOLP-Lasso& yes & & yes & \\
  HOLP-SCAD& yes& & &yes  \\
  HOLP-EBICS&yes&\\
  \hline
  Tilting& & & & \\
\hline
\end{tabular}
\end{center}
\end{table}

\end{document}